\documentclass[11pt]{article}

\usepackage{
  amssymb,
}
\usepackage[
  margin=1in,
  includefoot,
  footskip=30pt,
]{geometry}

\usepackage[font=small,labelfont=bf]{caption}

\usepackage{
  amsmath,
  amsthm,
  amsfonts,
  stmaryrd,
  url,
  pdfsync,
  graphicx,
  empheq,
  physics,
}

\usepackage[hidelinks]{hyperref}
\usepackage{cleveref}
\usepackage[table]{xcolor}

\newtheorem{theorem}{Theorem}

\newtheorem{corollary}{Corollary}
\newtheorem{definition}{Definition}

\newtheorem{lemma}{Lemma}
\newtheorem{invariant}{Invariant}
\newtheorem{notation}{Notation}

\newcommand{\concat}{\ensuremath{\sqcup}}

\newcommand{\gtext}[1]{\textcolor{gray!75}{#1}}
\newcommand{\algorithmlineheight}{\savebox\strutbox{$\vphantom{\highlight{y^\ell}}$}}
\newcommand{\reducedstrut}{\vrule width 0pt height .9\ht\strutbox depth .9\dp\strutbox\relax}
\newcommand{\highlight}[1]{%
  \begingroup
  \setlength{\fboxsep}{0pt}%
  \colorbox{red!20}{\reducedstrut$\displaystyle #1$\/}%
  \endgroup
}

\makeatletter
\DeclareRobustCommand{\var}[1]{\begingroup\newmcodes@\mathsf{#1}\endgroup}
\makeatother
\makeatletter
\DeclareRobustCommand{\keyword}[1]{\begingroup\newmcodes@{\bf \mathbf{#1}}\endgroup}
\makeatother
\newcommand{\defn}{\ensuremath{\triangleq}}

\providecommand{\keywords}[1]
{
  \small	
  \textbf{\textit{Keywords---}} #1
}

\title{Parallel RAM from Cyclic Circuits}
\author{David Heath\\
  \url{daheath@illinois.edu}\\
University of Illinois Urbana-Champaign}
\title{Parallel RAM from Cyclic Circuits}
\date{}

\begin{document}

\pagenumbering{gobble}

\maketitle

\begin{abstract}
Known simulations of random access machines (RAMs) or parallel RAMs (PRAMs) by
Boolean circuits incur significant polynomial blow-up,
due to the need to repeatedly simulate accesses to a large main memory.

Consider a single modification to Boolean circuits that
removes the restriction that circuit graphs are acyclic.
We call this the \emph{cyclic circuit} model.
Note, cyclic circuits remain \emph{combinational}, as they do not allow wire values to change over time.

We simulate PRAM with a cyclic circuit, and the blow-up from our simulation is only \emph{polylogarithmic}.
Consider a PRAM program $\mathcal{P}$ that on a length-$n$ input uses an arbitrary number of processors to manipulate words of size $\Theta(\log n)$ bits and then halts within $W(n)$ work.
We construct a size-$O(W(n)\cdot \log^4 n)$  cyclic circuit that simulates
$\mathcal{P}$.
Suppose that on a particular input, $\mathcal{P}$ halts in time $T$;
our circuit computes the same output within $T \cdot O(\log^3 n)$ gate delay.

This implies theoretical feasibility of powerful parallel machines.
Cyclic circuits can be implemented in hardware, and our circuit achieves
performance within polylog factors of PRAM.
Our simulated PRAM synchronizes processors via logical
dependencies between wires.
\end{abstract}
\keywords{Parallel Random Access Machine, Boolean Circuits, Circuit Cycles}

\newpage
\tableofcontents
\newpage

\pagenumbering{arabic}

\section{Introduction}\label{sec:intro}

The parallel random access machine (PRAM,~\cite{ForWyl78}) model is ubiquituous
in the study of parallel algorithms.
The model considers a single machine that hosts an arbitrary number of \emph{processors},
each of which operates in lockstep.
At each machine step, each processor is given read/write access to a
shared memory, and the machine satisfies these accesses simultaneously and in constant time.

We consider two measures of PRAM complexity:
the machine's \emph{work} measures the total number of instructions executed
across all processors, while the machine's \emph{runtime} (sometimes called its \emph{span}) measures the number
of machine steps.
Generally speaking, parallel algorithms reduce runtime while
keeping work reasonable, which often involves aggressively breaking a
problem into pieces and then using different processors to work on different pieces in parallel.

PRAM processors can easily distribute problems and combine solutions due to the model's \emph{synchronous} nature.
One processor can simply write a value to a shared memory cell with the
understanding that some other processor will read that cell on the next
machine step, achieving a communication channel.

Within PRAM there exists a hierarchy of variants that explain
what happens when more than one processor reads/writes the same memory cell in the same step.
The most powerful of these that is typically considered is concurrent-read concurrent-write (CRCW) PRAM,
which, as the name suggests, allows multiple processors to read/write the same memory cell without error.

While PRAM is central to the analysis of parallel algorithms, the model -- and in particular the CRCW variant --
is often seen as unrealistic.
It seems fundamentally challenging to build machines that implement
numerous tightly synchronized processors.
Today's multicore machines usually display only relatively
low degrees of parallelism, often limited to only dozens or hundreds of processors,
and parallelism is usually \emph{asynchronous}, with each processor
executing independently of the others.
Distributed systems can reach higher levels of parallelism, with the trade-off that synchronous behavior becomes even harder to achieve. 
This lack of synchronization makes it harder for processors to communicate, which complicates the system.
The inherent difficulty of implementing PRAM has led
to alternate parallel models that forgo PRAM's power and simplicity in exchange for more accurately
capturing real-world machines; see \Cref{sec:relwork}.
However, PRAM remains an attractive abstraction for reasoning about parallelism, and implementations of the model are highly desirable.

\subsection{Our Contribution: PRAM from Boolean Gates}

We show theoretical feasibility of machines that implement PRAM by simulating PRAM with a combinational Boolean circuit.
Our circuit design has size/delay that closely matches the work/runtime of the PRAM, and it can support any level of parallelism.

\paragraph{Our target model.}
We consider a PRAM that manipulates words of size $w = \Theta(\log n)$ bits.
The PRAM supports concurrent reads/writes (it is CRCW), and we consider
the most powerful accepted strategy for write conflict resolution:
when multiple processors write to the same address, the machine combines
written values with an arbitrary associative operation $\star$.

Our target PRAM supports varying \emph{processor activity}; at each step, any of the processors can be inactive.
The machine's Total work is the sum over the number of active processors in each machine step.
We \emph{do not} limit the maximum number of active processors.

See \Cref{sec:word-pram} for more details of the target model.

\paragraph{Our circuit's properties.}

We construct a combinational circuit $C$ that simulates the above PRAM. 
To construct $C$, we make only one meaningful change to the classically
accepted definition of a Boolean circuit: Boolean circuits are classically
\emph{acyclic}, and we remove this constraint.
This generalization of circuits has been
considered by prior work; see \Cref{sec:relwork}.
We emphasize that our circuit remains combinational in that its wires are \emph{stateless}:
each of $C$'s wire values is a (deterministic) function of $C$'s input.

Our main result is as follows:
\begin{theorem}[PRAM from Cyclic Circuits]\label{thm:main}
  Let $\mathcal{P}$ be a PRAM program that on length-$n$
  inputs halts within
  $W(n)$ work.
  There exists a cyclic circuit $C$ with $O(W(n) \cdot \log^4 n)$ fan-in-two gates such that for any length-$n$ input $\vb{x}$, $C_n(\vb{x}) = \mathcal{P}(\vb{x})$.
  Suppose that on input $\vb{x}$, $\mathcal{P}(\vb{x})$ halts in time $T$.
  Then $C_n(\vb{x})$ computes $\mathcal{P}(\vb{x})$ within $T \cdot O(\log^3 n)$ gate delay.
\end{theorem}

$C$ is essentially an \emph{ideal} PRAM, modulo polylog performance factors.
$C$ closely matches PRAM in terms of both work and runtime, and because $C$ is built from Boolean gates, it is theoretically feasible that we could build $C$.
Informally speaking, the constant factors in $C$'s size and delay are also
reasonable, as we \emph{do not} use concretely expensive components such as expander graphs~\cite{AKS83}.



\Cref{thm:main} -- together with a relatively obvious simulation of cyclic circuits by PRAM (\Cref{thm:reverse}) -- implies that cyclic circuits and PRAM are nearly \emph{equivalent} in power.
Indeed, the complexity classes implied by the two models are closely related:

\begin{corollary}[PRAM and Cyclic Circuits -- informal]\label{cor:pram-and-ckts}
  Denote by $\var{PRAM}(W(n), T(n))$ the set of problems solvable by a
  bounded-word-size PRAM (\Cref{sec:word-pram}) within $O(W(n))$ work and $O(T(n))$ time.
  Denote by $\var{CCKT}(W(n), T(n))$ the set of problems solvable by a 
  cyclic circuit family $\{ C_n:n \in \mathbb{N} \}$ where $C_n$ has $O(W(n))$ size and computes its output within $O(T(n))$ delay.
  For all $W(n)$ and for all $T(n)$:
  \begin{align*}
    \var{PRAM}&(W(n) \cdot \mathrm{poly}(\log n), T(n) \cdot \mathrm{poly}(\log n))\\
    =\var{CCKT}&(W(n) \cdot \mathrm{poly}(\log n), T(n) \cdot \mathrm{poly}(\log n))
  \end{align*}
\end{corollary}
(More formally, the above circuit family requires a modest notion of uniformity; see \Cref{sec:ckt}.)
Hence, our work shows that to develop parallel algorithms that are efficient in terms of both work and time (modulo polylog factors), one can work with PRAM, or one can work with Boolean gates. Choose whichever is more convenient.

\subsection{Intuition and Novelty}

At a high level, our combinational circuit design consists of a large number of small \emph{compute units}, each of which is a Boolean subcircuit that executes some single PRAM task, such as reading a memory element, adding two words, branching to a different instruction, etc.
The main challenge of our simulation is in properly gluing the units together.
Specifically, our simulation addresses the following key problems: 
\begin{itemize}
  \item How can we maintain a main memory such that each compute unit can access an arbitrary memory cell at amortized polylog gate cost?
  \item How can we coordinate compute units such that they can run in sequence or in parallel?
\end{itemize}
The main ingredients used to solve these problems are \emph{permutation networks}, which we use as the glue between compute units.
Prior works also glue PRAM with permutation networks, but the novelty of our work is in showing that \emph{the entire simulation} can be achieved by a single combinational Boolean circuit.

Simulating PRAM with just Boolean gates requires careful handling.
In particular,
(1) we ensure our permutation networks satisfy a \emph{dynamic} property, which  prevents compute units from `blocking' execution, and
(2) we `program' compute units to upgrade from the limited memory routing of a permutation to full-fledged CRCW parallel random access memory.
\Cref{sec:overview} explains our approach at a high level.

To our knowledge, this polylogarithmic connection between PRAM and combinational circuits has been overlooked.
We believe that this may be due in large part to a common misconception that combinational circuits \emph{must} be acyclic.
This is simply false, and Boolean circuits with cycles are well defined, without introducing notions from sequential hardware design such as stateful wires and timing; see examples throughout this work and discussion in~\cite{RieBru12}.
Even with the insight that cycles are well defined and useful, careful handling is required, and this work describes such handling in detail.

\subsection{Related Work}\label{sec:relwork}

\paragraph{PRAM and Circuits.}
\cite{StoVis84} was the first to simulate PRAM with
a poly-size Boolean circuit.
Their simulation incurs significant polynomial overhead.
They consider a PRAM with $n$-word input where each data word is at most $n$ bits long.
Let $T$ bound the runtime of the PRAM and let $p$ denote the number of processors.
Let $L = O(n + T + \log p)$.
\cite{StoVis84} simulate their considered PRAM by a circuit with
$O(p  T  L  (L^2 + p T))$ unbounded-fan-in Boolean gates.

\cite{StoVis84}'s considered machine manipulates larger words than ours, but
this discrepancy does not account for the difference in circuit size.
The main cost of \cite{StoVis84}'s simulation is that the circuit $T$ times simulates each of $p$ processors scanning each element in a shared memory of size $O(p \cdot T)$.
This immediately imposes $\Omega(p^2\cdot T^2)$ cost, even before word size is considered.
The high cost of acyclic circuit simulation seems inevitable, as it is hard to imagine a correct acyclic circuit that \emph{does not} scan all of shared memory on each simulated instruction, and to our knowledge no follow-on work improved the simulation.
While \cite{StoVis84}'s size blow-up is high, they use unbounded fan-in gates to achieve \emph{constant} blow-up in terms of circuit depth, establishing an important connection between circuits and PRAM.
\cite{KarRam90} discusses this connection further, in the context of the complexity class NC.

Our PRAM simulation uses fan-in-two gates and incurs only polylog overhead in terms of both size and delay.
Thus, our work can be understood as trading in polylog circuit delay in exchange for dramatically improved size.
To achieve this size improvement, we must consider circuit cycles.

\paragraph{Cyclic Circuits.}
Combinational circuits with cycles have been studied in numerous works, e.g.~\cite{Short60,Rivest77,Malik93,ShiBerTou96,Riedel04,MenShiBer12,RieBru12} and more.
These works use cycles to reduce circuit size,
but they obtain only relatively small benefit.
For instance, \cite{RieBru12} demonstrate a circuit family for which they can improve by only approximately factor two.
To our knowledge there are no prior results (aside from \cite{C:HeaKolOst23}, see next) that demonstrate \emph{asymptotic} improvement from the inclusion of cycles.
Our construction uses cycles to reduce PRAM overhead from a significant polynomial factor~\cite{StoVis84} to a relatively low \emph{polylogarithmic} factor.
Thus, we feel that our result shows cyclic circuits are far more interesting than previously thought.

Note, our formalization of cyclic circuits is similar to that of some of these prior works, in particular to that of~\cite{RieBru12} and \cite{MenShiBer12}.

To our knowledge, the connection between RAM and circuit cycles went unnoticed until the work of~\cite{C:HeaKolOst23}.
\cite{C:HeaKolOst23} formalized a cyclic circuit model called \emph{tri-state circuits} (TSCs).
The differences between TSCs and the cyclic circuits considered here are
not important in this work; we prefer cyclic Boolean circuits here because the gates are more
familiar.
Similar to our work, \cite{C:HeaKolOst23} demonstrate circuits that implement (non-parallel) RAM.
Both our approach and \cite{C:HeaKolOst23}'s approach
leverage \emph{dynamic permutation networks}, a key ingredient for simulating memory accesses.
\cite{C:HeaKolOst23} show that for a word RAM program running for $T = O(\mathrm{poly}(n))$ steps, there exists a (randomized) TSC with $O(T \cdot \log^3 T \cdot \log \log T)$ total gates that simulates the word RAM.
Our work follows on from \cite{C:HeaKolOst23}, taking the interesting and non-trivial step from RAM
to PRAM.

\paragraph{Connections to Cryptography.}
\cite{C:HeaKolOst23} used their TSC-based construction in the context of a cryptographic technique called circuit garbling~\cite{Yao86}.
By applying TSCs, \cite{C:HeaKolOst23} obtained a new approach to `Garbled RAM'~\cite{EC:LuOst13}.
This application required~\cite{C:HeaKolOst23} to consider \emph{oblivious} TSCs.
Oblivious TSCs are roughly analogous to oblivious Turing Machines~\cite{PipFis79}.
In short, circuits with cycles allow gates to propagate values in
data-dependent orders; in an oblivious TSC, this order must
`appear' (in a cryptographic sense) independent of the circuit's input.

Extending our PRAM results to oblivious execution could enable
interesting results for `oblivious parallel garbled
RAM'~\cite{TCC:BoyChuPas16}.
We do not pursue this result here, and achieving cyclic-circuit-based oblivious PRAM without unacceptably high overhead remains an interesting open problem.

\paragraph{Simulating PRAM; alternative parallel models.}

We are, of course, not the first to specify a construction that simulates PRAM.
However, a key difference between our work and prior works is that most
prior works take as primitive the notion of a standalone processor.
Namely, they start from the assumption that there is some collection of
processors, and what needs to be achieved is an appropriate coordination
of those processors, see e.g.~\cite{Vis84,Val90}.

Achieving PRAM starting from independent processors is
challenging, because inherently asynchronous processors must be synchronized.
This challenge
has led researchers to propose numerous alternative parallel models,
e.g. the aggregate model~\cite{DymCoo89},
the Bulk Synchronous Parallel model~\cite{Val90b},
LogP~\cite{CulKarPat93},
multhreaded models~\cite{AcaBleBlu00},
the Massively Parallel Communication model~\cite{BeaKouSuc17},
the fork-join model and its variants~\cite{BleFinGuSun20, CLRS22},
and more~\cite{ZhaCheSunMia07}.
The common thread of such models is that they embrace the asynchronous nature of independently executing processors.
While these models are important for
understanding how to control large and/or distributed
systems,
PRAM arguably remains the central ingredient in the study of
parallel algorithms, see e.g.~\cite{GibRyt89, CasLegRob08}.

Our approach circumvents the difficulty of processor synchronization.
For us, synchronous behavior is `free' in the sense that we
coordinate simulated processors with simple logical dependencies between
wires.
This automatically enforces a notion of lockstep behavior between processors, without the need for any additional enforcement.

As an aside, we technically introduce a parallel model, which we call
the parallel single access machine (PSAM) model; see \Cref{sec:psam}.
PSAM is a relatively natural weakening of PRAM.
We introduce PSAM simply as an intermediate step of our PRAM simulation, but the model may be of independent interest.
To our knowledge, this weakening of PRAM has not been explored.

\paragraph{Sorting networks and permutations networks.}

The key ingredient in our PRAM is a dynamically routing \emph{permutation network}.
Permutation networks are the subject of many works, see e.g. the classic works of~\cite{Ben64,Wak68,Bat68}.

\cite{CheChe96} presented a \emph{self-routing} permutation network using $O(n \cdot \log^2 n)$ swap operations.
Their network is based on a \emph{binary radix sort}.
Our construction also features a self-routing permutation network whose
structure is similar to that of \cite{CheChe96}.
Our network
(1) can be constructed from Boolean gates and
(2) permutes \emph{dynamically}, by which we mean that even if inputs arrive one by one, each input can be routed to its destination \emph{before} the next input arrives.
\cite{CheChe96} showed their network automatically routes all input packets to output destinations, but they do not show their network achieves the above \emph{dynamic} property where packets pass through the network even when not all packets are available.

\cite{C:HeaKolOst23} also constructed their circuit-based RAM from a
permutation network, but our network has better constants, and -- as we will see -- it
supports parallelism.
\cite{C:HeaKolOst23}'s network supports sequential routing only.

\section{Overview}\label{sec:overview}

This section sketches the main components of our approach at a high level.
Subsequent sections formalize the ideas explained here.

\subsection{Simulating Random Access with Gates}\label{sec:overview-memory}
Consider the Boolean basis of AND gates, XOR gates, and a
distinguished constant wire holding $1$.
Now,
(1) modify this basis such that
circuits are allowed to have cycles and
(2) ensure AND gates output zero eagerly.
Namely, an AND gate where one input wire holds zero outputs zero, regardless of the content of the second input wire.
Circuit cycles allow us to run certain parts of circuits before some input wires
to those parts are computed.
As we will see, this unlocks the ability to execute subcircuits
in \emph{data-dependent orders}.
By leveraging this, we can simulate PRAM with only a quasilinear number of gates.

\paragraph{Running gates in data-dependent orders.}
We start with an example that demonstrates this key data-dependence.
To explain our example, we need a helper component called a multiplexer.
A multiplexer is a standard circuit component that selects between two bits $x$ and $y$ based on some selection bit $s$:
\begin{center}
\begin{minipage}{0.35\textwidth}
\begin{center}
  \includegraphics{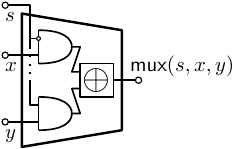}
\end{center}
\end{minipage}
\begin{minipage}{0.35\textwidth}
\begin{center}
\begin{align*}
  \var{mux}(s, x, y) \defn (1 \oplus s) \cdot x \oplus s \cdot y = \begin{cases}
    x & \text{if } s = 0\\
    y & \text{if } s = 1
  \end{cases}
\end{align*}
\end{center}
\end{minipage}
\end{center}
Because of AND's eager semantics, the output wire of the MUX \emph{does not depend} on whichever input is not selected.
For instance, if $s = 0$, the output does not depend on $y$, so the circuit can compute the multiplexer output \emph{before} it computes $y$.
In other words, this multiplexer inherits the eager behavior of AND.

Now, consider the following cyclic circuit on the left\footnote{
  This left example was noticed also in prior works, e.g. \cite{Stok92,Malik93}.
} which combines three multiplexers and
two unspecified subcircuits $f$ and $g$:
\begin{center}
  \includegraphics{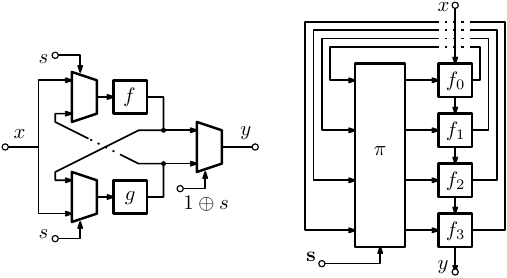}
\end{center}
Suppose we set the left circuit's selection bit $s$ to zero.
In this case, the two multiplexers on the left select their top argument, and
the multiplexer on the right selects its bottom argument.
Due to the eager behavior of multiplexers, the top left multiplexer outputs $x$, even though
its bottom input is not yet computed.
$x$ flows into $f$, and $f(x)$ flows to the bottom left multiplexer, which passes $f(x)$ to $g$.
Thus, the final multiplexer -- which outputs its bottom input -- outputs $g(f(x))$.

If we instead set $s$ to one, then the top left multiplexer
initially cannot fire, but the bottom left multiplexer can.
By tracing execution, we see that the circuit outputs $f(g(x))$.
Thus, the circuit computes the following:
\begin{align}\label{eqn:choice}
  y = \begin{cases}
    g(f(x)) & \text{if } s = 0\\
    f(g(x)) & \text{if } s = 1
  \end{cases}
\end{align}
Our example uses $f$ and $g$ in an order that depends on the runtime value $s$.
This demonstrates clear advantage over acyclic
circuits, because an acyclic circuit cannot in general compute the above
function, unless we include an extra copy either of $f$ or of $g$.

\paragraph{Connecting subcircuits with permutations.}
We can generalize our example to more than two subcircuits via a \emph{permutation
network}.
A permutation network is a circuit that \emph{routes} each of its inputs to a distinct output.
The routing of a permutation network can be chosen by runtime values,
so if we connect each network output to the input of some subcircuit, and if we cyclically connect the output of each subcircuit to a network input,
then we can compose an arbitrary number of subcircuits in a data-dependent order.
Importantly, permutation networks can be built from only a quasilinear number
of gates; see \Cref{sec:circuit-constructions}.

This generalization becomes interesting when we add more connections between
subcircuits; see the above right-hand example.
In this example, we connect subcircuits via a permutation $\pi$, and we also
\emph{sequentially} connect the subcircuits.
By properly setting up the subcircuits, we can arrange that $f_0,...,f_3$ run sequentially,
and each $f_i$ can \emph{write} a wire value by sending it to $\pi$.

By properly choosing $\pi$'s programming string $\vb{s}$, we can arrange that
$\pi$ routes $f_i$'s written value to some other subcircuit $f_{j \neq i}$,
allowing $f_j$ to \emph{read} the write.
$f_j$'s output might (or might not, especially if the write is made by a future subcircuit) depend on the read value.
This \emph{almost} simulates a memory access:
we write values to memory cells (wires), and then read them at arbitrary later
points in the evaluation.

\paragraph{From permutations to memory.}
There are two gaps between our example and full-fledged memory.
The first gap is that in our example the permutation network routing decisions $\vb{s}$ are made \emph{globally} and independently of the execution of the subcircuits.
To achieve random access memory, each subcircuit must \emph{locally} select its own read address, requiring
that the routing of the permutation network be chosen on the fly, as the subcircuits run.
The second gap is that our example is limited in that each memory cell can only be read \emph{once}, due to our use of a permutation.
We discuss resolution of this second gap in \Cref{sec:overview-unit}.

To resolve the first gap,
\Cref{sec:circuit-constructions} constructs a permutation network which is \emph{dynamically programmed}.
The network uses typical butterfly network configurations of swap elements, and it is similar in structure to existing networks, e.g.~\cite{CheChe96,Bat68}. We show that these familiar
structures can be made to work properly in the context of a cyclic circuit.

Our emphasis in \Cref{sec:circuit-constructions} is in showing our network's crucial \emph{dynamic} property: in our network, each routing decision can be made based only on a \emph{prefix} of inputs to the network.
This dynamic property is merely a result of the order of dependencies between network wires: the routing of the $i$-th input simply does not depend on the routing of any subsequent input $j > i$.
Thus, our network can route each of its inputs eagerly, before subsequent inputs to the network are computed.
This allows each of our subcircuits to locally choose its own read address while
preventing our simulation from becoming `deadlocked' with two subcircuits each waiting on the choice of the other.

\subsection{Dynamic Parallelism}

The above discussion sketches how cyclic circuits can achieve sequential RAM, but
our goal is to construct \emph{parallel} RAM.
We consider a PRAM whose number of active processors can vary over the program execution.
Because of this, PRAM runtime can also vary.
Our cyclic circuit faithfully simulates this varying parallelism, matching runtime performance in its \emph{gate delay}, up to polylog factors.
If the simulated PRAM is highly parallel, then the circuit has low delay; if the PRAM is highly sequential, then the circuit has higher delay.

The following example circuit on the left illustrates how varying delay can be achieved:
\begin{center}
\begin{minipage}{0.35\textwidth}
\begin{center}
  \includegraphics{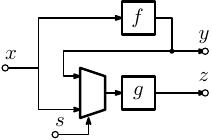}
\end{center}
\end{minipage}
\begin{minipage}{0.35\textwidth}
\begin{align*}
  (y, z) = \begin{cases}
    (f(x), g(f(x))) & \text{if } s = 0\\
    (f(x), g(x)) & \text{if } s = 1
  \end{cases}
\end{align*}
\end{minipage}
\end{center}

This circuit computes the function on the right.
If we set $s$ to zero, then the wire $z$ is a function of the \emph{sequential} composition of $f$ and $g$;
if we instead set $s$ to one, then $y$ and $z$ are functions of the \emph{parallel} composition of $f$ and $g$.
Due to eager semantics, the circuit has lower delay if it runs in
parallel than if it runs in sequence.
Namely, in the latter case, $g$ simply does not depend on the output of $f$, so $g$'s gates can start running even before $f$ is fully computed.

Note the distinction between circuit \emph{depth} and circuit \emph{delay}.
The latter accounts only for the longest chain of \emph{dependencies} through the circuit, and delay is the relevant runtime metric for cyclic circuits.

\subsection{Our PRAM Circuit}\label{sec:overview-sketch}

\begin{figure}[t]
  \begin{center}
  \includegraphics{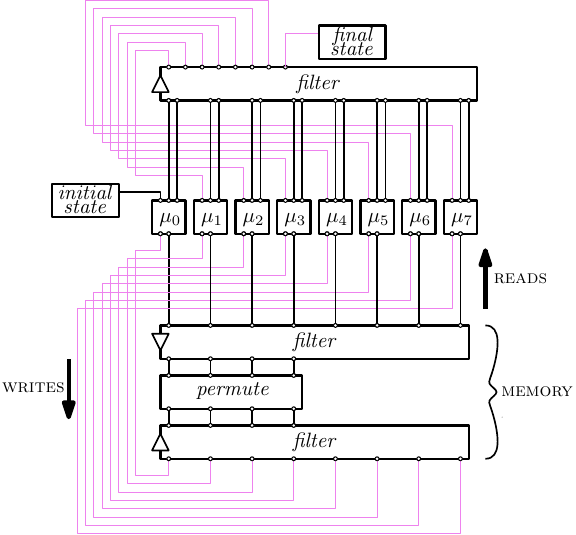}
  \end{center}
  \caption{
    Sketch of our circuit-based PRAM.
    The circuit has five main components:
    (1) a collection of compute units $\mu_i$,
    (2) a permutation network and two filters that jointly implement memory,
    (3) a filter that acts as a parallel coordination mechanism between compute units,
    (4) an input tape behind a filter (not depicted), and
    (5) an output tape behind a filter (not depicted).
  }\label{fig:pram}
\end{figure}

\Cref{fig:pram} depicts the high level structure of our PRAM-simulating circuit.
Our circuit focuses on the execution of \emph{compute units} $\mu_i$.
Each compute unit is a small circuit (its size is polylog in $n$) which can be understood as a handler for some single arbitrary basic
PRAM task, such as running (part of) one PRAM instruction.

PRAM execution starts at unit $\mu_0$, and it proceeds to the right through subsequent units.
For instance, if the PRAM behaves fully sequentially, then $\mu_0$ will compute some small task, then pass its state to $\mu_1$, which will compute some subsequent task, and so on.
The complexity of our construction comes from the \emph{coordination} of compute units.
We must arrange that (1) units can read/write a large shared memory and (2) units can run \emph{in parallel}.

\paragraph{Filters.}
To coordinate units, we apply permutation-like circuits called \emph{filters}:
\begin{center}
  \includegraphics{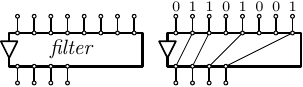}
\end{center}
A filter implements a routing between $n$ \emph{sources} and $n/2$ \emph{targets}.
Half of the sources should be tagged with zero, half should be tagged with one; those sources tagged with one are connected to the targets (preserving order),
and we `filter out' those sources tagged by zero.
\Cref{sec:circuit-constructions} constructs a dynamically routing filter with quasilinear size and logarithmic delay.

The role of filters in our PRAM is that they allow our compute units to \emph{opt
in} to particular machine capabilities.
For instance, our machine places its input tape behind a filter.
By sending a $1$ to the filter, a compute unit indicates that it would like to read a word of input, and the filter responds with the next input word; by sending a $0$, the unit indicates it does not need an input tape word, and the filter replies with an all-zero string.
Requiring units to explicitly opt in/out of each machine capability is tedious, but necessary.
This requirement ensures that early compute units do not ``block'' the operation of later units, since we can use the explicit opt-out to appropriately set up logical dependencies with Boolean gates.

Our filters require that precisely half of the units opt in and the other half opt out.
This requirement is not hard to achieve, but it requires that we use additional compute units which ``clean up'' machine execution by consuming unused resources.

\paragraph{Enabling dynamic parallelism via a filter.}
One important difference between \Cref{fig:pram} and our earlier example from \Cref{sec:overview-memory} is that in \Cref{fig:pram} we \emph{do not} connect compute units $\mu_i$ to each other directly.
Instead, connections pass through a filter (top of \Cref{fig:pram}).
This filter acts as a coordination mechanism, and its operation is key to our dynamic parallelism.
A compute unit can send a message through the filter, and this message is routed to some subsequent compute unit.
This allows compute units to pass state to successor units, activating those successors and continuing the computation.

Notice that each compute unit $\mu_i$ is connected to the
source side of the coordination filter \emph{twice}.
This is crucial. By connecting each unit twice, we allow unit $\mu_i$ to activate \emph{up to two} children.
$\mu_i$ can decide dynamically how many children it will activate by tagging modified copies of its state with $0$ or $1$ before sending them to the filter.
Think of $\mu_i$ as representing a single execution step of some parallel process.
$\mu_i$'s number of children represents three possible continuations of this process:
\begin{itemize}
  \item Zero children: The process terminates.
  \item One child: The process continues (in parallel with other processes).
  \item Two children: The process \emph{forks} into two parallel processes.
\end{itemize}
A program can quickly increase parallelism by having each unit activate two children.
For example,
\begin{itemize}
  \item $\mu_0$ sends two states through the filter, activating $\mu_1$ and $\mu_2$ in parallel.
  \item $\mu_1$ and $\mu_2$ each send two states through the filter.
    $\mu_1$'s states arrive at $\mu_3$ and $\mu_4$, and $\mu_2$'s states arrive at $\mu_5$ and $\mu_6$.
  \item $\mu_3,...,\mu_6$ each send two states through the filter, and so on.
\end{itemize}
Note that our circuit is fully synchronous and deterministic.
For instance, in this particular execution, $\mu_1$'s children are \emph{not} chosen arbitrarily.
Each unit $\mu_i$ will always have children of the lowest possible
index, with priority given to parents with lower indices (i.e., priority is given to \emph{older} processes).
This determinism comes simply from the fact that our construction is indeed a Boolean
circuit, albeit with cycles.

Parallelism in our circuit arises from the low delay of its components.
For instance, the coordination filter has logarithmic delay.
Hence, even if some huge number of units `simultaneously' request children, all
requests are handled within log delay.

\subsection{Programming Compute Units}\label{sec:overview-unit}

\begin{figure}
  \centering
  \includegraphics{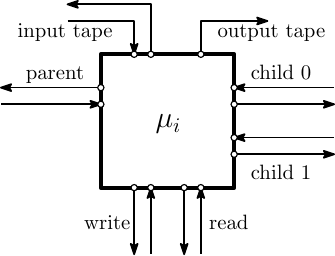}
  \caption{
    Structure of compute units $\mu_i$.
    Each unit activates when it receives a local state from its parent.
    It then performs some combination of the following:
    write a value to memory,
    read a value from memory,
    activate up to two children,
    read a word from the input tape,
    write a word to the output tape,
    compute some function of local state.
    By properly programming compute units, we achieve PRAM.
  }\label{fig:unit}
\end{figure}

By construction, our compute units have access to the various capabilities of our machine, including reading/writing shared memory, reading the input tape, writing the output tape, activating children, and responding to the parent; see \Cref{fig:unit}.
These capabilities are sufficient to implement our target PRAM.

However, appropriately `programming' these compute units $\mu_i$ is nontrivial.
The challenge here is that our circuit implements memory via a \emph{permutation}.
Each unit can write a memory element by sending it to the permutation network, but because we use a permutation, only \emph{one} unit can read that write.
In other words, our circuit's memory cells are inherently \emph{single-use}.
Of course, full-fledged PRAM allows repeated reads/writes to each memory address, and we must account for this discrepancy in our simulation.

\paragraph{The PSAM model.}
Our observation is that while we cannot use compute units to \emph{directly} run PRAM instructions, we \emph{can} use compute units to run simple parallel programs that manipulate binary-tree-based data structures.
Single-use memory cells are sufficient for this task because we can store pointers to binary tree child nodes inside parent nodes; each time we read a tree node from single-use memory, we can write back a fresh copy of that node.
With some care, we can use a tree-based program to implement PRAM.

We formalize the ability to manipulate binary trees by introducing a model that we call the parallel single access machine (PSAM) model.
In short, this model is the same as PRAM (\Cref{sec:word-pram}), except that each memory address can only be written to once and read from once.
See \Cref{sec:psam} for details.

Thus, we decompose our simulation into two parts.
\begin{itemize}
  \item First, we show that cyclic circuits can simulate PSAM. Each of our compute units simulates a single PSAM instruction, and we glue our units with permutations and filters.
  \item Second, we show that PSAM can simulate PRAM. Our PSAM maintains binary trees that store PRAM data, and by repeatedly traversing the trees, our PSAM simulates PRAM.
\end{itemize}
Plugging these together yields our contribution.

\paragraph{Our PSAM program.}
Our PRAM simulation uses circuit-based compute units to run a small PSAM
program that manipulates two trees.
The first \emph{memory tree} holds at its leaves all words that have been written in the PRAM shared memory;
the second \emph{processor tree} holds at its leaves the state of each active PRAM processor.
The high level goal of our PSAM program is to in parallel traverse the two trees together, matching up the state of each processor with the memory element
that processor wishes to access.
By doing so, we can take one step in that processor's execution.
By repeatedly traversing and rebuilding these two trees, we simulate a full PRAM execution.

In more detail,
the memory tree is a log-depth binary tree where each leaf along some path $i$ encodes the memory value written to address $i$.
The processor tree is also a log-depth binary tree, and each of its leaves stores the local state of some active PRAM processor.
The processor tree is arranged such that each processor's state \emph{aligns} with the memory address that processor wishes to access.
We sketch an example:
\begin{center}
  \includegraphics{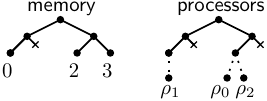}
\end{center}
Here, we consider a small memory with three addresses and three processors
$\rho_0, \rho_1, \rho_2$.
In our example, $\rho_1$ wishes to access memory address $0$ and $\rho_0, \rho_2$ each wish to access address $2$.
We ensure that $\rho_1$'s state is in the subtree rooted at position $0$, and $\rho_0$'s and $\rho_2$'s states are each in the subtree rooted at position $2$.
Our memory tree does not have an entry for address $1$ because no processor has yet accessed that address.
We implicitly store the all zeros string in each unaccessed address.

Our PSAM program implements a single PRAM machine step via a recursive procedure.
The PSAM procedure is roughly as follows:
\begin{itemize}
  \item Simultaneously and recursively traverse the memory tree and the processor tree.
    When the current processor tree node has two children, the current PSAM processor forks execution, and we continue in parallel down both branches of the memory/processor tree.
  \item Suppose we reach a memory leaf storing value $x_i$ (or an empty tree, denoting implicitly that $x_i = 0$).
    Save $x_i$ in the PSAM processor's local storage, then continue traversing the processor tree.
    When the current processor tree node has two children, fork execution,
    broadcasting $x_i$ to each child.
  \item Ultimately, the PSAM processor arrives at a leaf storing some PRAM processor state $\rho$.
    We compute a single PRAM instruction based on state $\rho$ and value $x_i$.
    The instruction writes\footnote{
      If the PRAM processor is merely reading a memory value, it can write back whatever it just read.
    } back some memory value $x_i'$.
    We encode the written value into a sparse tree with exactly one leaf, and this leaf is stored on path $i$.
    Additionally, the instruction can create zero, one, or two subsequent PRAM processor states.
    We encode these states into a tree with zero, one, or two leaves,
    each rooted at the memory address they wish to access next.
    Thus, we create a fresh memory tree and a fresh processor tree.
  \item
    The recursion begins to unwind.
    As it unwinds, we \emph{merge} together memory trees (resp. processor trees) created by recursive calls.
    Consider a merge on trees $t_0, t_1$.
    Our merge operation $\uplus$ ensures that each leaf node in $t_0$ (resp. $t_1$) appears on the same path in the merged tree $t_0 \uplus t_1$ as it does in $t_0$ (resp. $t_1$).
    When two merged trees share some leaf position, we combine those leaves with some binary associative operator $\star$.
    We can ensure that no two merged processor trees share leaves, so only memory tree merges use $\star$, and this is how we resolve write conflicts (see \Cref{sec:word-pram}).
    Computing $t_0 \uplus t_1$ is \emph{efficient} because we ensure that $t_0, t_1$ are sparse and/or $t_0,t_1$ share few leaf positions.
\end{itemize}

When the recursion completes, we are left with a fresh memory tree and a fresh processor tree.
The fresh memory tree stores the combined writes from each processor, as well untouched content from the initial memory tree;
the fresh processor tree stores all updated active processor states, and those states are rooted at the memory address they wish to access next.
Thus, we are back in a state accepted by our procedure, and we can apply the procedure again to implement another machine step.
Each call to this procedure takes one step in each processor, and it runs in $O(\log n)$ time.
By calling the procedure until the processor tree is empty, we simulate an end-to-end PRAM execution.

\subsection{Our Asymptotics}

Recall from \Cref{thm:main} that our circuit has size $O(W(n) \cdot \log^4 n)$ and delay $T \cdot O(\log^3 n)$.
We explain the sources of these factors.

Our circuit simulates $W(n)$ PRAM instructions, which in turn requires that it simulate $O(W(n) \cdot \log n)$ PSAM instructions.
This log factor loss comes from the fact that our PSAM program repeatedly traverses log-depth binary trees (\Cref{sec:overview-unit}).
Hence, each PRAM instruction is simulated by $O(\log n)$ PSAM instructions.

These PSAM instructions are, in turn, handled by $O(W(n) \cdot \log n)$ compute units (\Cref{fig:unit}).
Recall that these units must read/write memory elements, and this leads to the bottleneck in our circuit's size.
In particular, we connect these units by instantiating a permutation network, and to permute $m$ inputs our network uses $O(m \cdot \log^2 m)$ swap operations.
Because the network handles $\Theta(\log n)$-bit words, each swap requires $\Theta(\log n)$ Boolean gates.
Plugging in $m = O(W(n) \cdot \log n)$ as the number of permuted inputs, the size of the network is $O(W(n) \cdot \log^4 n)$ Boolean gates, dominating our circuit's size.

In terms of delay, our circuit incurs cost from
(1) the $O(\log^2 n)$ delay of our permutation network,
(2) the $O(\log n)$ overhead of the simulation of PRAM by PSAM, and
(3) the inherent $T$ steps required by the PRAM program itself.
Combining these results in our $T \cdot O(\log^3 n)$ total delay.

Thus, our cyclic-circuit-based simulation indeed achieves performance within polylog factors of the target PRAM.
The following sections expand on discussion given here, presenting our circuits and our simulation in full detail.

\section{Preliminaries}\label{sec:prelims}

\subsection{Word Parallel RAM}\label{sec:word-pram}

Our target model is a CRCW PRAM with bounded word size.
The PRAM allows processors that vary in number across program steps, and the PRAM combines write conflicts with an associative operator $\star$.
The following explains in detail.

\paragraph{Terminology.}
The PRAM's input length is denoted $n$.
The PRAM manipulates words of size $w = \Theta(\log n)$ bits.
PRAM input is stored on a tape;
when a processor reads the input tape, that word is popped from the tape such that when another processor reads the input tape, it obtains the next input word.
We similarly store the PRAM output on an output tape.

We place a modest bound on the PRAM's addressable main memory:
the highest memory address is at most polynomial in $n$.
This ensures that (1) our log-length words can address the shared memory and (2) the shared memory can be neatly encoded as a binary tree with logarithmic depth.

We place no bound on the machine's maximum number of processors $p$.
Each PRAM processor can be either \emph{active} or \emph{inactive}.
We refer to the number of steps for which a processor has been consecutively active as its \emph{age}.
If processor $\rho_0$ activated processor $\rho_1$, then we call $\rho_0$ `parent' and $\rho_1$ `child'.
If a processor goes inactive, it no longer has a parent (until it is activated again).

Each processor runs the same constant-sized program with instructions indexed by natural numbers, though different processors can run different instructions on the same step (the model is MIMD).
Each processor has a small\footnote{Our circuit cannot handle large processor local state without harming its asymptotics.
This is not a serious limitation, as processors can store local state in shared memory.}
 local state storing $O(1)$ words. Local state contains one distinguished word called the \emph{program counter}.
The program counter indicates which instruction to run next.

\paragraph{Complexity.}
The machine's \emph{runtime} $T$ denotes the total number of machine steps before the program halts.
The machine's \emph{work} $W$ denotes the total number of instructions executed across all processors before the machine halts.

\paragraph{Syntax and semantics.}
At each step, each active processor runs an instruction according to its program counter.
Instructions are chosen from the following grammar:
\begin{align*}
  \var{instr}
  ::=~&\vb{x} \gets f(\vb{x})   & \text{update local state} \\
    |~&y \gets \keyword{read}~x & \text{read address $x$; save the result in $y$} \\
    |~&\keyword{write}~x~ y     & \text{write word $y$ to address $x$} \\
    |~&x \gets \keyword{input}  & \text{read one word from the input tape} \\
    |~&\keyword{output}~x       & \text{write one word to the output tape}\\
    |~&\keyword{fork}~f(\vb{x}) & \text{activate a processor with local state $f(\vb{x})$}\\
    |~&\keyword{die}            & \text{go inactive}
\end{align*}
Above, $\vb{x}$ refers to the processor's entire local state, and metavariables $x, y$ refer to individual words in the local state.
Metavariable $f$ ranges over arbitrary functions that transform local state;
$f$ must be expressible as a polylog-uniform (in $n$) cyclic circuit (see \Cref{sec:ckt}) with $O(\log^2 n)$ gates and $O(\log^2 n)$ delay.
This is sufficient for addition, subtraction, comparison, multiplication, etc.
Notably, $f$ may manipulate the program counter, allowing conditional branching.

Machine execution begins with a single active processor in an all zeros local state, and where the shared memory stores all zeros.
If every processor is inactive, the machine halts.

\paragraph{Conflict resolution.}
When more than one processor writes the same address, the machine aggregates written values using an associative operation $\star$.
$\star$ can be instantiated by a polylog-uniform (in $n$) cyclic circuit with at most $O(\log^2 n)$ gates and $O(\log^2 n)$ delay.
This is sufficient to aggregate by, e.g., adding, multiplying, taking the first, taking the maximum, etc.

Since $\star$ is not necessarily commutative, the order in which the machine combines values matters.
Similarly, multiple processors might simultaneously read from the input tape/write to the output tape.
The machine resolves such conflicts according to processor age, where older processors receive priority;
ties are broken by the age of the processor's parent at the time the child was activated, then by the age of the grandparents, and so on.
Since machine execution starts with one processor, and since each processor can only fork one child at a time, this resolution is unambiguous.

When two or more processors read the input tape in the same step, the processor with the highest priority pops the first remaining word of the tape, the processor with the second-highest priority pops the next word, and so on.
Writing to the output tape is handled in the same manner, with higher priority processor output appearing first.

\subsection{Notation}

\begin{itemize}
  \item All logarithms are base two.
  \item Vectors are written in bold: $\vb{x}$.
  \item Vectors are indexed using bracket notation: $\vb{x}[i]$.  Indexing starts at zero.
  \item $\vb{x}[i..]$ denotes the subvector of $\vb{x}$ starting from index $i$.
    $\vb{x}[i..j]$ denote the subvector of $\vb{x}$ starting from index $i$ and ending at index $j$, inclusive.
  \item $[~]$ denotes an empty vector and $[x]$ denotes a singleton vector holding $x$.
  \item $\vb{x} \concat \vb{y}$ denotes the concatenation of $\vb{x}$ and $\vb{y}$.
  \item $x \defn y$ denotes that $x$ is equal to $y$ by definition.
  \item `msb' stands for `most significant bit'; `lsb' stands for `least significant bit'.
  \item We view index zero as the msb of a vector, as it is furthest to the left.
\end{itemize}

\section{Cyclic Circuits}\label{sec:ckt}

We simulate PRAM with a cyclic circuit.
This section formalizes cyclic circuits, and we explain the semantics and complexity measures of the model.

\subsection{Syntax and Semantics}

For concreteness, we choose a particular gate set, allowing AND/XOR circuits with cycles:
\begin{definition}[Cyclic Circuit]\label{defn:ckt}
  A cyclic circuit $C$  is a circuit
  allowing cycles (i.e., its wiring graph need not be acyclic)
  composed from fan-in two AND/XOR gates.
  $C$ has $n$ input wires and $m$ output wires.
  $C$ may use a distinguished wire $1$ that
  holds constant $1$.
  Each wire in $C$ has a distinct identifier from some set $\var{wire-id}$.
\end{definition}

The semantics of cyclic circuits are defined by stating which values can appear
on each circuit wire.
An \emph{assignment} (defined next) is a \emph{map} from wires to values such that assigned values satisfy constraints imposed by gates.

\begin{definition}[Assignment]\label{defn:assign}
  Let $C$ be a cyclic circuit and let $\vb{x} \in \{0,1\}^n$ be an input.
  An \textbf{assignment}  for $C, \vb{x}$ is a map
  $\var{asgn} : \var{wire-id} \rightarrow \{0, 1\}$
  that sends e of $C$'s wires to a value.
  An assignment $\var{asgn}$ is considered \textbf{valid} if
  (1) each $i$-th input wire is sent to corresponding input value $\vb{x}[i]$, and
  (2) the output wire of each gate $g$ is related to $g$'s input wires
  according to $g$'s function:
  \begin{center}
  \begin{minipage}[t]{0.3\textwidth}
  \begin{center}
    \begin{tabular}{c|cc}
      $\oplus$ & $0$ & $1$\\
      \hline
      $0$ &  $0$ & $1$\\
      $1$ &  $1$ & $0$
    \end{tabular}
  \end{center}
  \end{minipage}
  \quad
  \begin{minipage}[t]{0.3\textwidth}
  \begin{center}
    \begin{tabular}{c|cc}
      $\cdot$ & $0$ & $1$\\
      \hline
      $0$ & $0$ & $0$\\
      $1$ & $0$ & $1$
    \end{tabular}
  \end{center}
  \end{minipage}
  \end{center}
\end{definition} 
We emphasize that AND outputs zero if either of its arguments is zero.
This captures the eager nature of AND, enabling our constructions.

\paragraph{Legal circuits.}
Consider a simple circuit $C$ defined as follows:
\begin{align}\label{eqn:loop}
  & \keyword{let}~y = x \cdot y~\keyword{in}~y
\end{align}
So far, $C$ is well-defined with respect to \Cref{defn:ckt}:
it is a single AND gate whose second input wire is its own output:
\begin{center}
  \includegraphics{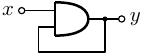}
\end{center}
However, this circuit is problematic.
Suppose we consider input $x = 1$.
The pair $C, x$ admits two valid assignments:
\begin{align*}
  \{x \mapsto 1 , y \mapsto 0\} \qquad
  \{x \mapsto 1 , y \mapsto 1\}
\end{align*}
Indeed, both settings of $y$ satisfy the AND gate constraint.

While it is possible to consider cyclic circuits with multiple valid assignments, it is far simpler to only consider those circuits that have exactly one assignment per input. We say that such circuits are \emph{legal}.
\begin{definition}[Legal Cyclic Circuit]\label{defn:legal}
  A cyclic circuit $C$ is considered \textbf{legal} if for any input $\vb{x} \in \{0, 1\}^n$, the pair $C, \vb{x}$ has exactly one valid assignment (\Cref{defn:assign}).
  If $C$ is not legal, it is considered \textbf{illegal}.
\end{definition}
Henceforth and outside this section, all considered cyclic circuits are legal.
For instance, \Cref{thm:main} refers to legal cyclic circuits, and the example circuits we considered in \Cref{sec:overview} were legal.

\begin{notation}[Wire values]
  When $C$ and $\vb{x}$ are clear from context, we denote the single
  assignment for $C, \vb{x}$ by $\var{val}$.
  We denote by $C(\vb{x})$ the string of values in the image of $\var{val}$ corresponding to $C$'s output wires.
\end{notation}

\subsection{Complexity Measures}

In standard \emph{acyclic} Boolean circuits, we typically measure circuit complexity via size and depth.
In \emph{cyclic} circuits, we instead measure size and \emph{delay}.
The size of a cyclic circuit $|C|$ is simply its number of gates.
The \emph{delay} of a wire measures the time needed before that wire acquires its value.
We assume that each gate takes unit time to propagate input to output.
Wire delay \emph{depends on circuit input}:
\begin{definition}[Wire Delay]\label{defn:delay}
  Let $C$ be a cyclic circuit with input $\vb{x} \in \{0, 1\}^n$.
  The \textbf{wire delay} of $C, \vb{x}$ is a map
  $\var{delay} : \var{wire-id} \rightarrow \mathbb{N}$
  that sends each wire to the lowest value
  satisfying the following constraints:
  \begin{align*}
    \var{delay}(w_0 \oplus w_1) &= 1+\var{max}(\var{delay}(w_0), \var{delay}(w_1))\\
    \var{delay}(w_0 \cdot w_1) &\geq 1+\var{min}(\var{delay}(w_0), \var{delay}(w_1))\\
    \var{delay}(w_0 \cdot w_1) &\geq 1+\var{delay}(w_0) & \text{if $\var{val}(w_1) \neq 0$}\\
    \var{delay}(w_0 \cdot w_1) &\geq 1+\var{delay}(w_1) & \text{if $\var{val}(w_0) \neq 0$}
  \end{align*}
\end{definition}
As an example, $\var{delay}$ maps each input to $0$, as this is the lowest natural number and as \Cref{defn:delay} places no further constraint.
The delay of an AND gate depends on its inputs, reflecting the gate's eager nature.
If an AND gate's faster input holds zero, then the gate has low delay; if not, the gate has higher delay, as it must wait until its slower input acquires its Boolean value.

Note that because each wire's value is uniquely determined by the circuit input (\Cref{defn:legal}), it is relatively straightforward that each wire must have finite delay.

\paragraph{Circuit delay.}
Formally, wire delay is a measure for \emph{wires}, but it is
also useful to discuss the delay of a \emph{circuit}.
The delay of $C$ with respect to input $\vb{x} \in \{0,1\}^n$ is the maximum delay of $C$'s wires.
We also consider $C$'s delay without a specified input $\vb{x}$, which is defined as the highest
delay over all inputs $\vb{x} \in \{0,1\}^n$.

\subsection{Uniformity}

We are interested not only in showing that cyclic circuit families can simulate PRAM, but also the other way around.
The former is far more interesting, but the latter is needed to tightly connect the two models.
To show a simulation of cyclic circuits by PRAM, we need a notion of uniformity such that a PRAM can efficiently compute the description of a given circuit.

\begin{notation}[Polylog Uniform/Computable]\label{notation:uniform}
  Say that a circuit family $\{ C_n:n \in \mathbb{N} \}$ is polylog uniform if -- upon input $n$ and $i$ -- a random access Turing machine can compute the description of $C_n$'s $i$-th gate in time
  $O(\mathrm{poly}(\log n))$.
  Similarly, say a quantity $f(n)$ is polylog computable if the quantity can be computed by a random access Turing machine in time $O(\mathrm{poly}(\log n))$.
\end{notation}

This uniformity is convenient for connecting PRAM and cyclic circuits, because it means PRAM can work-efficiently simulate cyclic circuits.
Note that the more standard logspace uniformity condition is not sufficient for our purposes,
because all we can conclude about a logspace program is that it must halt in \emph{polynomial} time, and our theorems cannot tolerate polynomial work blow-up.
Another standard notion is DLOGTIME uniformity, but this notion seems insufficient to describe our circuits without blowing up the circuit size by a polynomial factor.

We can now more formally state \Cref{cor:pram-and-ckts}.
\begingroup
\def\thecorollary{\ref{cor:pram-and-ckts}}
\begin{corollary}[PRAM and Cyclic Circuits]
  Denote by $\var{PRAM}(W(n), T(n))$ the set of problems solvable by a
  bounded-word-size PRAM (\Cref{sec:word-pram}) within $O(W(n))$ work and $O(T(n))$ time.
  Denote by $\var{CCKT}(W(n), T(n))$ the set of problems solvable by a 
  polylog-uniform cyclic circuit family $\{ C_n:n \in \mathbb{N} \}$ where $C_n$ has $O(W(n))$ size and computes its output within $O(T(n))$ delay.
  For $W(n) = O(\mathrm{poly}(n))$ s.t. $W(n)$ is polylog-computable, and for all $T(n)$:
  \begin{align*}
    \var{PRAM}&(W(n) \cdot \mathrm{poly}(\log n), T(n) \cdot \mathrm{poly}(\log n))\\
    =\var{CCKT}&(W(n) \cdot \mathrm{poly}(\log n), T(n) \cdot \mathrm{poly}(\log n))
  \end{align*}
\end{corollary}

\subsection{Simulating Cyclic Circuits with PRAM}\label{sec:sim-cckt}

The goal of this work is to simulate PRAM with a cyclic
circuit, but the other direction is also needed to establish a strong connection between the two models.
The following is a relatively straightforward fact:
\begin{theorem}[Cyclic Circuits from PRAM]\label{thm:reverse}
  Let $\{C_n : n \in \mathbb{N} \}$ be a polylog-uniform
  cyclic circuit family
  such that each $C_n$ has size at most $W(n) = O(\mathrm{poly}(n))$.
  There exists a PRAM program $\mathcal{P}$ such that for any length-$n$ input $\vb{x}$,
  $\mathcal{P}(\vb{x})$ outputs $C_n(\vb{x})$ within $O(W(n) \cdot \mathrm{poly}(\log n))$ work.
  If on a particular input $\vb{x}$, $C_n(\vb{x})$ outputs within $T$ delay, then $\mathcal{P}(\vb{x})$ halts within time at most $T \cdot O(\mathrm{poly}(\log n))$.
\end{theorem}
\begin{proof}
  Straightforward from the uniformity of the circuit family.

  First, $\mathcal{P}$ computes the description of $C_n$ in $O(\mathrm{poly}(\log n))$ time and $O(W(n) \cdot \mathrm{poly}(\log n))$ work.
  The description can be computed efficiently in parallel due to polylog-uniformity.
  
  Next, $\mathcal{P}$ simulates each gate.
  As its invariant, $\mathcal{P}$ maintains a set of processors, each of which
  represents a wire whose Boolean value has already been computed.
  $\mathcal{P}$ sets up this invariant by assigning one processor to each circuit input wire.
  Then, each wire processor handles fan-out by
  (recursively) forking children such that each connected gate has a corresponding processor.
  This processor attempts to evaluate the gate.
  If the gate's output is indeed determined by the current configuration of its
  input wires, then the processor marks the gate as handled and becomes the handler of the gate's output wire.
  (To avoid two processors handling the same output wire, the processors use gate markers to decide who takes control of the output wire; this is possible due to PRAM's synchronous nature.)
  If the gate's output is not yet determined, the processor simply saves its
  wire value on the gate input and goes inactive.

  By handling wire fan-out in a binary tree fashion, every wire value is propagated
  to each connected gate in $O(\log n)$ time, and every gate is touched only twice (once per gate input).
  Thus, the total runtime is $T \cdot O(\mathrm{poly}(\log n))$ and the total work is $O(W(n) \cdot \mathrm{poly}(\log n))$.
\end{proof}

\section{Parallel Single Access Machines}\label{sec:psam}

There seems to be some natural tension between the PRAM model -- which allows arbitrary re-use of stored data values --
and the cyclic circuit model -- where each wire can only be read by statically connected gates.
We accordingly introduce an intermediate model: parallel single access machines (PSAMs).
PSAM seems to be more naturally compatible with cyclic circuits because each memory address can be written to/read from only once, a natural fit for our permutation-based memory.
We give two simulations:
\begin{itemize}
  \item We show cyclic circuits can simulate PSAM.
  \item We show PSAM can simulate PRAM.
\end{itemize}

The PSAM model is similar to bounded PRAM (\Cref{sec:word-pram}), with the notable exceptions that
(1) each memory address can be read \emph{only once}, and
(2) written memory addresses are chosen by the machine, not the program.
These restrictions are quite strong, but it remains possible to construct
programs that manipulate and traverse tree-like data structures.

\paragraph{Syntax and semantics.}

PSAM is identical to PRAM (\Cref{sec:word-pram}), except that we change
the instruction set.
PSAM instructions are specified by the following grammar:
\begin{align*}
  \var{instr}
  ::=~&\vb{x} \gets f(\vb{x})           & \text{update local state} \\
    |~&y \gets \keyword{read}~x         & \text{read value stored at address $x$} \\
    |~&y \gets \keyword{write}~x        & \text{write $x$ in a fresh address; save address in $y$} \\
    |~&x \gets \keyword{input}          & \text{read one word from the input tape} \\
    |~&\keyword{output}~x               & \text{write one word to the output tape}\\
    |~&y \gets \keyword{fork}~f(\vb{x}) & \text{activate a processor with local state $f(\vb{x})$}\\
    |~&\keyword{ret}~x                  & \text{halt current process with output $x$}
\end{align*}

The key differences between PSAM and PRAM are in the specification of
$\var{read}$/$\var{write}$ and of $\var{fork}$/$\var{ret}$.
\begin{itemize}
  \item On $\var{write}$, the processor does not choose the storage address.
    Instead, the \emph{machine} arbitrarily selects some fresh address, saves the written
    value at that address, and tells the processor which address was chosen.
    (Our PSAM-simulating circuit simply increments addresses starting from zero.)

  \item On $\var{read}$, the machine returns the content at the
    target address, then it \emph{invalidates} that address.
    If a processor reads an invalid address, then the machine `crashes'
    and immediately outputs $\bot$.
    Useful PSAM programs never read the same address more than once.

  \item After $\var{fork}$, the spawned child may return one word to its parent by calling $\var{ret}$.
    $\var{fork}$ sets aside one word of parent local memory $y$ as a \emph{pointer}
    to the child's returned word.
    If the parent \emph{dereferences} this pointer before the
    child returns, the machine `crashes' and immediately outputs $\bot$. 
    Useful PSAM programs do not dereference return pointers
    until after they are set by $\var{ret}$.
    Operating on a pointer, e.g. by using an update instruction to modify its
    content, dereferences the pointer.
    Saving a pointer in memory \emph{does not} dereference the pointer.
\end{itemize}

The $\var{fork}/\var{ret}$ pointer is needed because -- unlike in PRAM -- PSAM processors cannot communicate through shared memory alone.
Indeed, as addresses can only be written once, and as addresses are chosen arbitrarily,
a child cannot write a value that its parent can read, except by using the $\var{ret}$ return value.

In PSAM, we use a word-size $w = \Theta(\log n)$ where the hidden constant is
large enough to store two memory addresses, plus extra metadata, sufficient to
implement binary tree data structures.

\paragraph{PSAM capabilities; specifying PSAM programs.}
Rather than tediously writing out PSAM programs in the above instruction
format, we specify PSAM programs as simple recursive function definitions.
Our PSAM programs manipulate \emph{binary trees} that store data words at their leaves.
We use the following inductive definition:

\algorithmlineheight
\begin{empheq}[box=\fbox]{align}\label{defn:tree}
  \var{Tree}
  ::= \var{Empty}
  ~|~\var{Leaf}(\var{word})
  ~|~ \var{Branch}(\var{word},\var{Tree}, \var{word}, \var{Tree})
\end{empheq}
Each branch node stores two pointers to its subtrees, as well as two natural numbers denoting the
\emph{depth} of each subtree.

In the following, we argue that we can compile our recursive function definitions to PSAM programs. To aid understanding, we specify an example PSAM program that
in parallel sums 
words stored at the leaves of a binary tree:
\begin{align*}
  \gtext{1}~~&\var{sum}(t) \defn\\
  \gtext{2}~~&~~\keyword{match}~t~\keyword{with}\\
  \gtext{3}~~&~~~~\var{Empty} \mapsto 0\\
  \gtext{4}~~&~~~~\var{Leaf}(x) \mapsto x\\
  \gtext{5}~~&~~~~\var{Branch}(d^\ell, t^\ell, d^r, t^r) \mapsto\\
  \gtext{6}~~&~~~~~~\keyword{let}~(t_{\var{shallow}}, t_{\var{deep}}) \defn \keyword{if}~d^\ell \leq d^r~\keyword{then}~(t^\ell, t^r)~\keyword{else}~(t^r, t^\ell)\\
  \gtext{7}~~&~~~~~~~~~~~s_0 \defn \keyword{PAR}(\var{sum}(t_{\var{shallow}}))\\
  \gtext{8}~~&~~~~~~~~~~~s_1 \defn \var{sum}(t_{\var{deep}})\\
  \gtext{9}~~&~~~~~~\keyword{in}~s_0 + s_1
\end{align*}
In words, the sum of a tree's leaves is (1) zero if the tree is empty, (2) $x$ if the tree is a leaf storing word $x$, or (3) computed by summing subtrees in parallel, then adding the results.
Our specification delegates the sum of the shallower subtree to a child process; the reason for this is explained shortly. 

Such specifications map to PSAM programs in the following way:
\begin{itemize}
  \item \textbf{Tree structures.}
    Our specifications manipulate binary trees.
    The PSAM maintains trees by storing
    nodes in shared memory.
    Each node has tag bits that indicate its `type': $\var{Empty}$,
    $\var{Leaf}$, or $\var{Branch}$.
    Each node also stores relevant data; in the case of a
    $\var{Branch}$ node, subtree roots are stored as memory addresses.
    Our PSAM implements $\keyword{match}$ expressions on trees via $\var{read}$ instructions;
    in our example, the processor's first step is to read the root node $t$.
  \item \textbf{Case analysis.}
    Our specifications include \emph{conditional} behavior, e.g. depending on the type
    of tree node.
    A PSAM processor handles conditional behavior by updating its program
    counter depending on bits in the local state.
    In our example, the processor conditionally jumps to one of
    three program locations, depending on the type tag of $t$.
  \item \textbf{Basic expressions.}
    Our specifications include basic expressions, e.g. adding two values, comparing two values, etc. 
    A PSAM processor handles such expressions by using its ability to update local state with arbitrary circuit-expressible functions $f$.
  \item \textbf{Function invocation.}
    Our specifications include function calls.
    A PSAM processor handles function calls by maintaining a standard \emph{call stack}.
    I.e., before jumping to the code of the called function, the processor pushes its local
    state to a stack data structure, which is stored in shared memory.
    At all times, the processor maintains a pointer to the top of the stack.
    Because local state is only a constant number of words, we can save local state with only a constant number of $\var{write}$ instructions.
    To return from a called function, the processor uses $\var{read}$ to
    pop its old state and continue where it left off.
    In our example, the processor saves its local state before the call to
    $\var{sum}(t_{\var{deep}})$; after the recursive call concludes, the
    processor reads its old state before performing the final addition.
  \item \textbf{Parallel execution.}
    Our specifications include \emph{parallel} behavior, e.g. traversing two
    branches of a tree in parallel.
    We denote this by writing $\keyword{PAR}(e)$ where $e$ is an arbitrary
    expression to be computed by a child process.
    A PSAM processor handles parallel expressions by invoking $\var{fork}$.
    The value of $e$ is represented by the pointer returned by $\var{fork}$.
    \textbf{Important note:} it is \emph{crucial} that on $\keyword{PAR}(e)$,
    the child must return before the parent uses the value of the parallel
    expression; else the PSAM will dereference an invalid return pointer and crash.
    In our example, we delegate the \emph{shallower} tree to the
    child process, ensuring the child computes
    $\var{sum}(t_{\var{shallow}})$ using fewer instructions than the call to
    $\var{sum}(t_{\var{deep}})$.
    Because PSAM processors operate in lockstep, the child computes its sub-sum before the parent, so
    the subsequent addition $s_0 + s_1$ does not dereference an invalid pointer.
  \item \textbf{Destructive behavior.}
    When our specifications case analyze trees, those trees are destroyed.
    This is reflected in our PSAM by the fact that reading a memory element invalidates the read address.
    It is thus important to check that specifications do not use the same tree twice.
    Notice that in our example, $t$, $t^\ell$, $t^r$, $t_{\var{shallow}}$, and
    $t_{\var{deep}}$ are each used at most \emph{once}, regardless of the program's
    execution path.
    If we wished to adjust $\var{sum}$ such that $t$ is not `erased', we could
    rebuild $t$ as the recursion unwinds.
\end{itemize}

In \Cref{sec:circuit-constructions}, we simulate PSAM with Boolean gates alone.
The key ideas behind this simulation are depicted in \Cref{fig:pram}:
we implement main memory with a permutation network, and we implement each PSAM instruction with a compute unit.
These units are connected to each other through a filter, allowing for calls to $\var{fork}$/$\var{ret}$.

\section{PSAM from Cyclic Circuits}\label{sec:circuit-constructions}

In this section, we simulate PSAM (\Cref{sec:psam}) with a cyclic circuit (\Cref{sec:ckt}).
The goal of this section is to establish the following:

\begin{lemma}[PSAM from Cyclic Circuits]\label{thm:psam-ckt}
  Let $\mathcal{P}$  be a PSAM program (\Cref{sec:psam}) that on length-$n$
  inputs halts within $W(n)$ work, where $W(n) = O(\mathrm{poly}(n))$ and the
  quantity $W(n)$ is polylog-computable (\Cref{notation:uniform}).
  There exists a cyclic circuit $C_n$ of size $O(W(n) \cdot \log^3 n)$
  that simulates $\mathcal{P}$ on all length-$n$ inputs.
  Suppose that on length-$n$ input $\vb{x}$, $\mathcal{P}(\vb{x})$ halts in time $T$.
  Then $C_n(\vb{x})$ computes its output within delay $T \cdot O(\log^2 n)$.
  The family $\{C_n:n \in \mathbb{N}\}$ is polylog-uniform.
\end{lemma}
The proof of \Cref{thm:psam-ckt} is by construction of the circuit family $C_n$; the construction is described in the remainder of this section.
By combining this result with results from \Cref{sec:psam-pram}, we obtain our simulation of PRAM (\Cref{thm:main}).

\subsection{Swap}
Our construction uses a permutation network to route memory elements between subcircuits.
The primitive underlying this network is a \emph{swap}:
\[
  \var{swap}(s, x, y) \defn \left(
    (\neg s \cdot x) \lor (s \cdot y),
    (\neg s \cdot y) \lor (s \cdot x)\right)
\]
Here, $\lor$ denotes logical OR and $\neg$ denotes logical NOT; each of these can be implemented from AND/XOR/1.
$\var{swap}$ outputs $(x, y)$ when $s = 0$, and it outputs $(y, x)$ when $s = 1$.
For convenience, we generalize $\var{swap}$ to the following definition, which swaps two length-$w$ vectors:
\algorithmlineheight
\begin{empheq}[box=\fbox]{align*}
  \var{swap}_w(s, \vb{x}, \vb{y}) \defn \left(
    (\neg s \cdot\vb{x}) \lor (s \cdot \vb{y}),
    (\neg s \cdot\vb{y}) \lor (s \cdot \vb{x})\right)
\end{empheq}
This definition treats $\lor$ as the element-wise OR of two vectors and $\cdot$ as the AND scaling of each vector element by a single scalar.
$\var{swap}_w$ uses $O(w)$ total gates.

\paragraph{Eagerness of $\var{swap}$.}

While $\var{swap}$ may seem simple, its small definition belies a subtle detail.
To see this, we propose another strawman definition:
\[
  \var{bad}(s, x, y) \defn \left(
    s \cdot(x \oplus y) \oplus x, s \cdot(x \oplus y) \oplus y\right)
\]
This gate seems to compute the same function as $\var{swap}$.
One might even attempt to argue that $\var{bad}$ is superior to $\var{swap}$, as it can be
computed with fewer gates.
However, in the context of a cyclic circuit, $\var{swap}$ and $\var{bad}$ are \emph{not equivalent}.
To see this, suppose that input $y$ is not yet computed at the time we consider the swap, which we denote by setting $y$ to $\bot$:
\begin{center}
\begin{tabular}{ccc|cc}
  $s$ & $x$ & $y$ & $\var{swap}(s, x, y)$ & $\var{bad}(s, x, y)$\\
  \hline
  $0$ & $x$ & $\bot$ & $(x, \bot)$ & $(x, \bot)$ \\
  \rowcolor{red!20}$1$ & $x$ & $\bot$ & $(\bot, x)$ & $(\bot, \bot)$  \\
\end{tabular}
\end{center}
The table shows that $\var{swap}$ can \emph{eagerly} forward $x$ to its
output wires, even before the value of $y$ is known.
(Indeed, $\var{swap}$ also eagerly forwards $y$ if $x = \bot$.)
$\var{bad}$ cannot eagerly forward $x$ because its second output \emph{always} depends on $y$.
There are two important points to this example.

First, the rules of Boolean algebra do not necessarily apply in a cyclic circuit.
$\var{bad}$ is not equivalent to $\var{swap}$ because, in a cyclic circuit, AND does not distribute over XOR.
Replacing $x\cdot(y \oplus z)$ by $x\cdot y \oplus x \cdot z$ can \emph{change circuit dependencies}.
On the other hand, many Boolean algebra rules \emph{do} apply.
For instance, AND and XOR each remain commutative and associative.

Second, the eager nature of $\var{swap}$ is \emph{central} to our construction.
To see why, suppose that bit $x$ is computed by some step of the PSAM, and
suppose $y$ is computed by some \emph{later} step.
Indeed, suppose that $y$ \emph{depends on $x$}.
Even in this situation, $\var{swap}$ works, because it can deliver $x$ to the destination where $y$ is computed, then $y$ can be routed as input to $\var{swap}$ via a cycle.
It is precisely this eager feature of $\var{swap}$ that allows us to build a
single permutation network that wires together PSAM steps.

\subsection{Helper Circuits}\label{sec:helper}

We construct some subcircuits used in our PSAM simulation.

First, $\var{halves}$ takes a vector of wires $\vb{x}$ and splits it into two vectors of half the length.
$\var{halves}$ has no gates:
\algorithmlineheight
\begin{empheq}[box=\fbox]{align*}
  &\var{halves}(\vb{x}) \defn \keyword{let}~n \defn |\vb{x}|~\keyword{in}~(\vb{x}[0..n/2-1], \vb{x}[n/2..n-1])
\end{empheq}

Second, we formalize a classic ripple-carry adder.
\algorithmlineheight
\begin{empheq}[box=\fbox]{align*}
  \gtext{1}~~&\vb{x} + \vb{y} \defn\\
  \gtext{2}~~&~~\keyword{if}~|\vb{x}| = 0~\keyword{then}~[0]~\keyword{else}\\
  \gtext{3}~~&~~~~\keyword{let}~[\var{carry}] \concat \vb{lsbs} \defn \vb{x}[1..] + \vb{y}[1..]\\
  \gtext{4}~~&~~~~~~~~~\var{msb} \defn \vb{x}[0] \oplus \vb{y}[0] \oplus \var{carry}\\
  \gtext{5}~~&~~~~~~~~~\var{carry}' \defn (\vb{x}[0] \oplus \var{carry}) \cdot (\vb{y}[0] \oplus \var{carry}) \oplus \var{carry}\\
  \gtext{6}~~&~~~~\keyword{in}~[\var{carry}'] \concat [\var{msb}]\concat \vb{lsbs}
\end{empheq}
We fully specify this circuit
(1) for completeness
(2) to familiarize the reader with our circuit notation, and
(3) to emphasize a property of the adder's delay.

The adder is defined recursively.
To add two zero-bit values, we output the one-bit vector $[0]$.
To add $n$-bit values, we
(1) add together the $n-1$ low bits,
(2) strip the `carry' from the recursively computed sum,
(3) compute the new msb and carry, and
(4) concatenate the full $n+1$-bit sum.

Note that the adder outputs its lsb within $O(1)$ delay.
This is because the lsb of the sum depends only on the lsbs of the input.
Correspondingly, the $i$-th lsb of the adder output is computed within delay $O(i)$.
This property will be useful later when we argue that our future $\var{partition}$ circuit has low delay; this fact will rely on \emph{pipelining} adders, which is possible due to the above property.

\subsection{Permutations}

This section formalizes the main circuit components of our construction.
Our goal is to construct a \emph{dynamic permutation network}.
This network takes as input $n$ words, each tagged by some distinct target destination.
The network automatically routes each word to its destination within polylog delay.

The crucial \emph{dynamic} property of the network is that even if only some
\emph{prefix} of input words is available, those words are
\emph{eagerly routed} to their destination within polylog delay.
This eagerness is central to our handling of PSAM/PRAM, since it allows \emph{sequential} composition of instructions.
At the same time, the network's low delay enables efficient \emph{parallel} composition of instructions.

Our network is essentially a binary radix sort implemented in hardware, and it
is similar in structure to prior permutation/sorting networks~\cite{Bat68,CheChe96}.
Our emphasis is to show that the network indeed achieves dynamic routing.

\paragraph{Partition at position $i$.}

\begin{figure}
  \centering
  \includegraphics[width=0.65\textwidth]{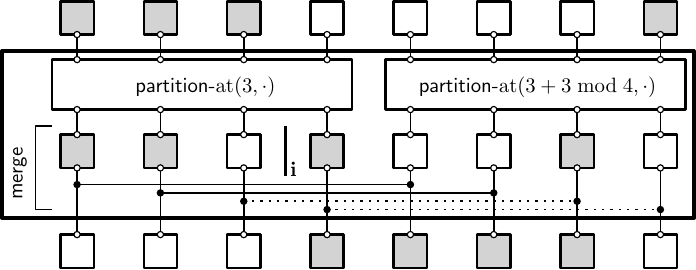}
  \caption{
    An example of $\var{partition-at}$.
    Inputs are depicted at the top; outputs are at the bottom.
    Consider eight elements, four of which have
    msb zero (shaded); the others have msb one
    (unshaded). 
    $\var{partition-at}$ rearranges elements such that shaded elements
    occur consecutively starting at position $i$ (here, $i= 3$).
    We first recursively partition each half of the array,
    and then
    $\var{merge}$ conditionally swaps elements, depending on their shading and on which side of $i$ they are on.
  }\label{fig:merge-example}
\end{figure}

The main component of our permutation network is a subcircuit that we call $\var{partition-at}$. $\var{partition-at}$ takes as input
(1) a vector of $n \cdot w$ wires $\vb{x}$, where $n$ is a power of two, and
(2) a vector of $\log n$ wires $\vb{i}$.
The circuit interprets $\vb{x}$ as an array of $n$ length-$w$ words, and it interprets $\vb{i}$ as an index of that array.
$\var{partition-at}$ uses $\var{swap}$ gates to \emph{rearrange} the content of $\vb{x}$ such that those words in $\vb{x}$ with msb $0$ occur to \emph{right} of position $i$, and those with msb $1$ occur to the \emph{left} of $i$, wrapping around as necessary.

\Cref{fig:merge-example} sketches an example; this sketch is useful throughout the following our explanation.
The $\var{partition-at}$ circuit is formalized below:
\algorithmlineheight
\begin{empheq}[box=\fbox]{align*}
  \gtext{1}~~&\var{partition-at}_{n,w}(\vb{i},\vb{x}) \defn\\
  \gtext{2}~~&~~\keyword{if}~n = 1~\keyword{then}~([1 \oplus\vb{x}[0]],\vb{x})\\
  \gtext{3}~~&~~\keyword{else}\\
  \gtext{4}~~&~~~~\keyword{let}~(\vb{x}^\ell,\vb{x}^r) \defn \var{halves}(\vb{x})\\
  \gtext{5}~~&~~~~~~~~~(\vb{zeros}^\ell, \vb{y}^\ell) \defn \var{partition-at}_{n/2,w}(\vb{i}[1..], \vb{x}^\ell)\\
  \gtext{6}~~&~~~~~~~~~(\vb{zeros}^r, \vb{y}^r) \defn \var{partition-at}_{n/2,w}((\vb{i} + \vb{zeros}^\ell)[2..], \vb{x}^r)\\
  \gtext{7}~~&~~~~~~~~~(\vb{z}^\ell, \vb{z}^r) \defn \var{merge}_{n,w}(\vb{i}[0], 0, 0, \vb{i}[1..], \vb{y^\ell}, \vb{y}^r)\\
  \gtext{8}~~&~~~~\keyword{in}~(\vb{zeros}^\ell + \vb{zeros}^r, \vb{z}^\ell \concat \vb{z}^r)
\end{empheq}

$\var{partition-at}$ recursively breaks down the array $\vb{x}$; as the recursion unwinds, it uses a linearly-sized, log-delay subcircuit called $\var{merge}$ to combine the two halves; $\var{merge}$ is explained later.

The key detail of $\var{partition-at}$ is its management of index $i$.
Recall, we wish to place all msb $0$ elements to the right of $i$.
Let $\vb{x}^\ell, \vb{x}^r$ denote the halves of $\vb{x}$.
$\var{partition-at}$ uses its first recursive call to place msb $0$ elements of $\vb{x}^\ell$ to the right of $i$ (actually, we place elements to the right of $i \bmod n/2$). 
The idea here is to \emph{align} $\vb{x}^\ell$ elements with their position in the output vector.

Next, we wish to place each msb $0$ element from $\vb{x}^r$ to the right of all msb $0$ elements from $\vb{x}^\ell$.
To achieve this, we shift $i$ to the right before making our second recursive call.
This explains the use of adders.
Each call to $\var{partition-at}$ actually achieves two tasks:
(1) it partitions the elements as already described and
(2) it counts the number of msb zero elements in $\vb{x}$.
The first recursive call thus tells us how many elements are already to the right of $i$, allowing us to position $\vb{x}^r$ elements further to the right.

\paragraph{Merging recursively partitioned arrays.}
$\var{partition-at}$ wishes to place zero elements to the right of position $i$, but $i$ is in the range $[0, n)$.
On its recursive calls, $\var{partition-at}$ adjusts $i$ such that each recursive $i$ is in the \emph{smaller} range $[0, n/2)$.
Suppose we were to take the two recursive output arrays $\vb{y}^\ell, \vb{y}^r$ and simply \emph{concatenate} them.
Because the recursive calls operate on arrays of length $n/2$, not length $n$, each element in the concatenated array could be in one of two possible array positions:
(1) the correct position or
(2) a position exactly distance $n/2$ from the correct position.

These two possibilities must be resolved, motivating the design of $\var{merge}$:
\algorithmlineheight
\begin{empheq}[box=\fbox]{align*}
  \gtext{1} ~~&\var{merge}_{n,w}(\var{parity}, \var{left}, \var{right}, \vb{i}, \vb{x}, \vb{y}) \defn\\
  \gtext{2} ~~&~~\keyword{if}~n=1~\keyword{then}\\
  \gtext{3} ~~&~~~~\keyword{let}~\var{s} \defn \var{parity} \oplus \var{left} \oplus \vb{x}[0]\\
  \gtext{4} ~~&~~~~\keyword{in}~\var{swap}_w(s, \vb{x}, \vb{y})\\
  \gtext{5} ~~&~~\keyword{else}\\
  \gtext{6} ~~&~~~~\keyword{let}~(\vb{x}^\ell, \vb{x}^r) \defn \var{halves}(\vb{x})\\
  \gtext{7} ~~&~~~~~~~~~(\vb{y}^\ell, \vb{y}^r) \defn \var{halves}(\vb{y})\\
  \gtext{8} ~~&~~~~~~~~~(\vb{z}^{\ell}, \vb{w}^{\ell}) \defn \var{merge}_{n,w}(\var{parity}, \vb{i}[1..], \var{left} \lor (\vb{i}[0] \cdot \neg \var{right}), \var{right}, \vb{x}^\ell, \vb{y}^\ell)\\
  \gtext{9} ~~&~~~~~~~~~(\vb{z}^r, \vb{w}^r) \defn \var{merge}_{n,w}(\var{parity}, \var{left}, \var{right} \lor (\neg \vb{i}[0] \cdot \neg \var{left})
  \vb{i}[1..], \vb{x}^r, \vb{y}^r)\\
  \gtext{10}~~&\mathrlap{~~~~\keyword{in}~(\vb{z}^\ell \concat \vb{z}^r, \vb{w}^\ell \concat \vb{w}^r)}
\end{empheq}
$\var{merge}$ combines two partitioned vectors $\vb{x}, \vb{y}$ from calls to $\var{partition-at}$.
The high-level operation of $\var{merge}$ is to element-wise swap entries of $\vb{x}$ and $\vb{y}$ yielding a single partitioned array.
$\var{merge}$ is conceptually a \emph{half-cleaner} from the sorting network literature.
The challenge of this operation is in deciding \emph{which} elements should be swapped and which should not.

$\var{merge}$ recursively breaks down $\vb{x}$ and $\vb{y}$;
once $\vb{x},\vb{y}$ each contain exactly one word, it conditionally swaps those two words.
The key detail of $\vb{merge}$ is its management of the position $i$ as well as variables $\var{left}$ and $\var{right}$.
The value of $\var{left}$ (resp. $\var{right}$) denotes an answer to the following question:
does the currently considered subvector of $\vb{x}$ lie entirely to the left (resp. right) of $\var{partition-at}$'s value of $i$?
These bits are useful because once we reach the base case, we know we wish to move the single element in $\vb{x}$ to the right depending on whether we are considering a location that is left of $i$.
The value $\var{parity}$ flips this logic if the original value of $i$ is at least $n/2$.

As the recursion unwinds, $\var{merge}$ concatenates results of its recursive calls,
in the end yielding two correctly partitioned halves, which
$\var{partition-at}$ concatenates together. 

\paragraph{Stability.}

The $\var{partition-at}$ circuit is \emph{stable} in the following sense.
Output words with msb $0$ appear in their original relative order;
output words with msb $1$ appear in the \emph{reverse} of their original relative order.
Namely, the output sequence is \emph{bitonic}~\cite{Bat68}.

The stability of $\var{partition-at}$ can be seen by an inductive argument.
In the base case, partitioning a singleton vector is trivially stable.
In the general case, we have two stable bitonic sequences $\vb{y}^\ell$ and $\vb{y}^r$, where each $0$-tagged $\vb{y}^r$ elements appear to the \emph{right} of all $0$-tagged $\vb{y}^\ell$ elements (modulo $n/2$), and similarly for all $1$-tagged elements;
by merging, we thus ensure that all $0$-tagged elements appear in their original relative order, and similarly for all $1$-tagged elements.
Thus, $\var{partition-at}$ indeed achieves this notion of bitonic stability.

\paragraph{Dynamic behavior.}
Crucially, $\var{merge}$'s base case decision of whether to swap two elements
is only made with respect to (1) the index $i$ and (2) the msb of $\vb{x}$.
I.e., the value of $\vb{y}$ is \emph{irrelevant}.
Moreover, if we revisit $\var{partition-at}$, we notice that arguments to the
first recursive call \emph{do not depend} on the second recursive call.
These two parts together mean that $\var{partition-at}$ starts sending elements to correct
positions even when only an arbitrary prefix of $\vb{x}$ is available.

These properties are crucial to the dynamic behavior of our permutation network,
because they mean that decisions about how to route a word depend
only on the values of words originally to that element's \emph{left}.
Rephrasing this in the context of PSAM/PRAM, we can correctly route a memory element based only on
that element's target address and \emph{prior} routing decisions.

\paragraph{The partition circuit and its complexity.}

\begin{figure}
  \centering
  \includegraphics[width=0.4\textwidth]{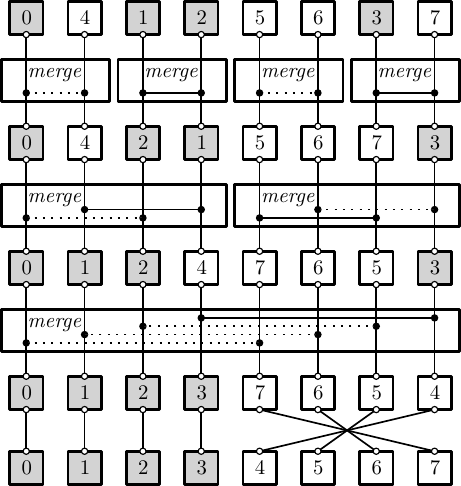}
  \qquad
  \includegraphics[width=0.4\textwidth]{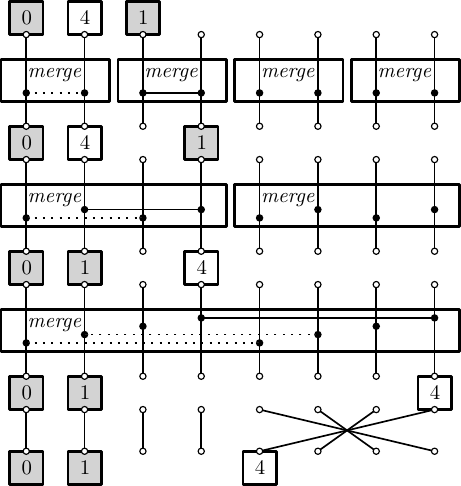}
  \caption{
    (Left) An example partition.
    Inputs are depicted at the top; outputs are at the bottom.
    The partition is \emph{stable}: it preserves the
    relative order of elements.
    Horizontal bars depict calls to $\var{swap}$.
    (Right) The circuit begins routing even when the input is only partially available.
    This dynamic routing is possible because the destination of each 
    element depends only on the destination of elements initially to the left.
  }\label{fig:partition-example}
\end{figure}

We provide a simple wrapper around $\var{partition-at}$ which
(1) partitions at position zero,
(2) reverses the second half of the vector,
and (3) drops the msb of each word in the output.
Dropping the msb is convenient for composing partitions into a full permutation.
\algorithmlineheight
\begin{empheq}[box=\fbox]{align*}
  \gtext{1}~~&~\var{partition}_{n,w}(\vb{x}) \defn\\
  \gtext{2}~~&~~\keyword{let}~(\text{\textvisiblespace}~, \vb{z}) \defn \var{partition-at}_{n,w}(0^{\log n}, \vb{x})\\
  \gtext{3}~~&~~~~~~~(\vb{z}^\ell, \vb{z}^r) \defn \var{halves}(\vb{z})\\
  \gtext{4}~~&~~\keyword{in}~(\var{drop-msbs}_w(\vb{z}^\ell), \var{drop-msbs}_w(\var{reverse}(\vb{z}^r)))
\end{empheq}
Above, $\var{drop-msbs}$ denotes a procedure that drops the msb of each length-$w$ word in the argument vector.
\Cref{fig:partition-example} depicts an end-to-end call to $\var{partition}$.

$\var{partition}$ has the following complexity:
\begin{itemize}
  \item $\var{partition}$ has size at most $w\cdot O(n \cdot \log n)$.
  \item $\var{partition}$ has delay $O(\log n)$.
\end{itemize}
  The size of $\var{partition}$ can be derived from
  (1) solutions to basic recurrence equations,
  (2) the linear size of ripple-carry adders, and
  (3) the fact that $\var{swap}_w$ has $\Theta(w)$ gates.

  Maximum delay is more nuanced.
  Indeed, the $\var{partition-at}$ circuit involves adders on $O(\log n)$-bit strings, and adders have linear delay.
  Thus, it may \emph{seem} that $\var{partition-at}$ has $O(\log^2 n)$ delay, caused by a sequence of $O(\log n)$ adders on $O(\log n)$ bit numbers.
  However, recall from discussion in \Cref{sec:helper} that an adder produces its low bits with low delay and high bits with higher delay.
  Ultimately, this \emph{pipelines} the adders induced by $\var{partition-at}$'s recursion:
  the adders at the leaves produce their lowest bits within constant delay,
  allowing adders one level up the recursion tree to compute their low bits only a
  constant delay later, and so on.
  In this manner, all adders complete in $O(\log n)$ delay. 

\paragraph{Permutation.}

\begin{figure}
  \centering
  \includegraphics[width=0.4\textwidth]{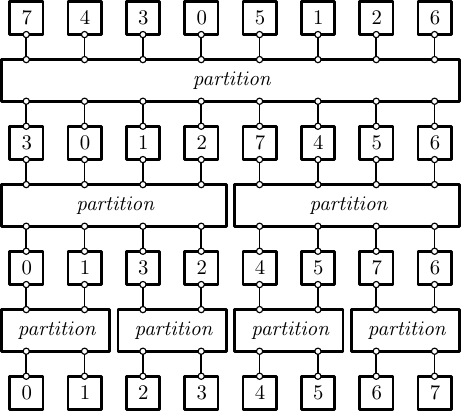}
  \qquad
  \includegraphics[width=0.4\textwidth]{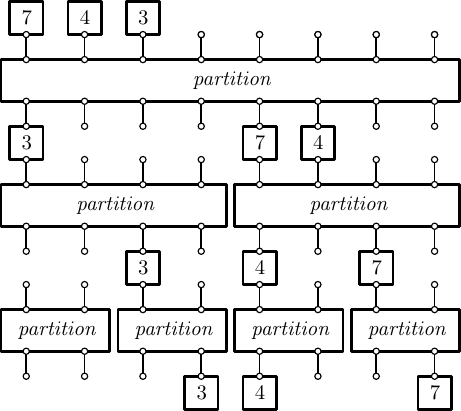}
  \caption{
    (Left) An example permutation.
    Inputs are depicted at the top; outputs are at the bottom.
    The circuit implements a binary radix sorting network; it repeatedly
    partitions input elements based on bits of each
    element's destination.
    (Right) The circuit begins routing even when only a prefix of input wires are set:
    the permutation is \emph{dynamic}.
    This dynamic nature is inherited from the dynamic nature of our partition network; see \Cref{fig:partition-example}.
  }\label{fig:permutation-example}
\end{figure}

We achieve a dynamic permutation network by recursively applying our partition network (see \Cref{fig:permutation-example}):
\algorithmlineheight
\begin{empheq}[box=\fbox]{align*}
  \gtext{1}~~&\var{permute}_{n,w}(\vb{x}) \defn\\
  \gtext{2}~~&~~\keyword{if}~n = 1~\keyword{then}~\vb{x}\\
  \gtext{3}~~&~~\keyword{else}\\
  \gtext{4}~~&~~~~\keyword{let}~(\vb{x}^0, \vb{x}^1) \defn \var{partition}_{n, \log n + w}(\vb{x})\\
  \gtext{5}~~&~~~~\keyword{in}~\var{permute}_{n/2, w}(\vb{x}^0)\concat\var{permute}_{n/2, w}(\vb{x}^1)
\end{empheq}
$\var{permute}$ takes as input $2^k$ words $\vb{x}$ where each entry
of $\vb{x}$ is tagged with a distinct target location.
It uses $\var{partition}$ to split $\vb{x}$ in half such that the resulting left vector $\vb{x}^0$ contains those elements intended for locations with msb $0$, and $\vb{x}^1$ contains those elements intended for locations with msb $1$.
Recall that $\var{partition}$ drops the msb of each vector input, so the msb of each tag is dropped, allowing us to simply recursively apply $\var{permute}$ to each half.
$\var{permute}$ calls $\var{partition}$ with word size $\log n + w$ to account for the $\log n$ tag bits.
Permuting a singleton vector is trivial.
$\var{permute}$ can be understood as an implementation of binary radix sort where all inputs are distinct.

$\var{permute}$'s size/delay can be calculated using the complexity of $\var{partition}$ and by solving simple recurrences:
\begin{itemize}
  \item \textbf{Size.} $\var{permute}$ has size at most $w\cdot O(n \cdot \log^2 n)$.
      When $w = \Theta(\log n)$, $\var{permute}$ has $O(n \cdot \log^3 n)$ gates.
  \item \textbf{Delay.} $\var{permute}$ has $O(\log^2 n)$ delay.
\end{itemize}

\paragraph{Filters.}

\Cref{sec:overview-sketch} describes a \emph{filter} subcircuit, which takes as input $2^k$ words and keeps those $2^{k-1}$ words tagged with $1$.
$\var{filter}$ can be built from a single partition:
\algorithmlineheight
\begin{align*}
  \var{filter}_{n,w}(\vb{x}) \defn
  \keyword{let}~(\text{\textvisiblespace}, \var{keep}) \defn \var{partition}_{n,w+1}(\vb{x})~
  \keyword{in}~\var{keep}
\end{align*}
$\var{filter}$ inherits $\var{partition}$'s size and delay.
$\var{filter}$ calls $\var{partition}$ with word size $w+1$ to account for the single extra tag bit.

\paragraph{Bidirectional permutations and filters.}

Our permutation network takes $n$ tagged input words and sends those
words to $n$ distinct output locations.
So far, our network is \emph{one-directional}, in the sense that it sends data from its source side to its target side only.
To achieve RAM we need a \emph{bidirectional} permutation network that allows data to flow both from source to target and from target back to source.
Each source \emph{request} flows through the network in one direction and
the corresponding \emph{response} flows back in the opposite direction.
This is needed for memory lookups, where a compute unit inputs a memory address to the source side of a network. From here, the network connects this read to the appropriate memory write, and then the written value flows back through the network to the source side where it was requested.

Suppose we have $n$ source addreses and $n$ (untagged) \emph{target} inputs.
Our $\var{bipermute}$ circuit sends each source address to the addressed location, and then pulls the corresponding target element back to the source.
$\var{bipermute}$ is almost identical to $\var{permute}$; just rotate some $\var{swap}$ components such that they point from target to source instead of source to target.
The asymptotic size and delay of this subcircuit thus matches that of $\var{permute}$.
Our formal circuit $\var{bipermute}_{n,w}$ takes as input two vectors:
(1) a length-$n$ vector of $(\log n)$-bit tagged source-side inputs and
(2) a length-$n$ vector of $w$-bit target-side inputs.
It outputs a length-$n$ vector of $w$-bit source-side outputs.
As $\var{bipermute}$ is a simple extension of $\var{permute}$, we do not specify further.

We assume a similar generalization of $\var{filter}$.
Like $\var{filter}$, $\var{bifilter}$ connects half of $n$ sources with $n/2$ targets.
$\var{bifilter}$ simply extends $\var{filter}$ such that these connections are bidirectional.
Formally, $\var{bifilter}_{n,w,\omega}$ takes as input two vectors:
(1) a length-$n$ vector of $(1 + w)$-bit tagged source-side inputs and
(2) a length-$n/2$ vector of $\omega$-bit target-side inputs.
It outputs
(1) a length-$n$ vector of $\omega$-bit source-side outputs (each source tagged with $0$ receives an all zeros output) and
(2) a length-$n/2$ vector of $w$-bit target-side outputs.

\subsection{Simulating PSAM}
\begin{figure}[t]
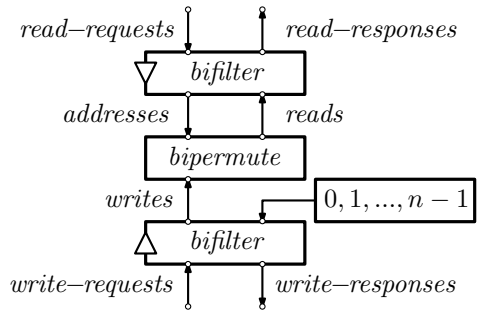


  {\small
\algorithmlineheight
\begin{empheq}[box=\fbox]{align*}
  \gtext{1}~~&\var{memory}_{n,w}(\var{write-reqs}, \var{read-reqs}) \defn\\
  \gtext{2}~~&~~\keyword{let}~(\var{read-resps}, \var{addresses}) \defn \var{bifilter}_{n, \log n-1, w}(\var{read-reqs}, \var{reads})\\
  \gtext{3}~~&~~~~~~~\var{reads} \defn \var{bipermute}_{n/2, w}(\var{addresses}, \var{writes})\\
  \gtext{4}~~&~~~~~~~(\var{write-resps}, \var{writes}) \defn \var{bifilter}_{n, w, \log n-1}(\var{write-reqs}, [0,1,...,n/2 - 1])\\
  \gtext{5}~~&~~~~\keyword{in}~(\var{write-resps}, \var{read-resps})
\end{empheq}

\algorithmlineheight
\begin{empheq}[box=\fbox]{align*}
  \gtext{1}~~&\var{PSAM}_{\mu,n,w}(\var{input-tape}) \defn\\
  \gtext{2}~~&~~\keyword{let}~(\var{input-reqs},\var{output-reqs},\var{write-reqs},\var{read-reqs},\var{to-parent},\var{to-children}) \defn\\
  \gtext{3}~~&~~~~~~~~~~\mu^{n}(\var{input-resps},\var{write-resps},\var{read-resps},\var{from-parent},\var{from-children})\\
  \gtext{4}~~&~~~~~(\var{write-resps}, \var{read-resps}) = \var{memory}_{n/2, w}(\var{write-reqs}, \var{read-reqs})\\
  \gtext{5}~~&~~~~~(\var{input-resps}, \text{\textvisiblespace}) = \var{bifilter}_{n,0,w}(\var{input-reqs}, \var{input-tape})\\
  \gtext{6}~~&~~~~~\var{output-tape} = \var{filter}_{n,w}(\var{output-reqs})\\
  \gtext{7}~~&~~~~~(\var{from-children}, \var{from-parent}) = \var{bifilter}_{n,w,w}(\var{to-children}, \var{to-parent})\\
  \gtext{8}~~&~~\keyword{in}~\var{output-tape}
\end{empheq}
}

\caption{
  Our circuit-based memory unit (top) 
    connects $n$ conditional writes with $n$ conditional reads (exactly $n/2$ of each pass through a respective filter).
    The memory responds to each write with the address where the entry was written. 
  Our PSAM circuit (bottom) instantiates $n$ computes units $\mu$ (see \Cref{fig:unit}) and connects them to each other, to the memory, and to the input/output tapes.
  Memory words and messages sent between compute units are each of size $w = \Theta(\log n)$.
  We denote by $\mu^n$ the parallel composition of $n$ copies of $\mu$.
  Each vector input to $\mu^n$ is
  partitioned into $n$ chunks with each chunk being sent to one
  compute unit $\mu$; similarly outputs of each compute unit are concatenated into the resulting six vectors.
}\label{fig:psam}

\end{figure}

We are now ready to construct our full PSAM-simulating circuit.
Note, \Cref{fig:pram} is a relatively faithful depiction of our full circuit.
The circuit itself, along with a circuit that implements our PSAM memory, is listed in \Cref{fig:psam}.

First, our circuit includes $O(W(n))$ \emph{compute units}, where $W(n)$ is the total work of the simulated PSAM program.
We have not yet described the content of these units, but at its interface a compute unit has six input ports and seven output ports (see \Cref{fig:unit}).
Each of these ports are $O(w)$ bits wide.
To complete our construction, we must properly connect up these ports.

To do so, we first instantiate a \emph{memory unit} which consists of one $\var{bipermute}$ subcircuit and two $\var{bifilter}$ subcircuits.
The $\var{bipermute}$ subcircuit is placed between the two bifilters; we arrange that each $\var{bifilter}$ has an equal number of target ports as the $\var{bipermute}$ subcircuit has source/target ports:
\begin{center}
\includegraphics{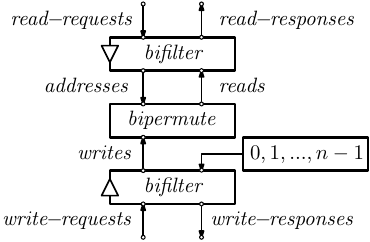}
\end{center}
This memory unit responds to read requests on its source side and to write requests on its target side.
It handles these requests by using the permutation to match reads with writes.
Each write request receives as response the next available address.

In addition to our memory, our PSAM circuit includes an input tape and an output tape, each behind a filter.
Finally, we connect our units to one another through a coordination bifilter, allowing each unit two activate up to two subsequent units, and allowing child units to send a word back to their parent.

In sum, we connect compute units to an infrastructure that allows each unit to leverage some combination of machine capabilities, which we next list.
In this list, `tag' denotes a leading $\{0, 1\}$ value that indicates whether a
payload should pass through a filter or not.
Each compute unit can perform some combination of the following:
\begin{itemize}
  \item Send one tag bit through a bifilter to conditionally read one input tape word.
  \item Send a tagged word through a bifilter to conditionally write the output tape.
  \item Send two tagged words through a bifilter to activate up to two children and receive responses.
  \item Send one word in response to the parent.
  \item Send one tagged word to memory to conditionally write a word.
  \item Send one tagged address to memory to conditionally read a word.
\end{itemize}
The precise handling done by units is described next.

\subsection{Compute Units}

We have now set up the full infrastructure in which we embed our compute units.
The remaining task is to say what these compute units actually \emph{are}.

Recall that our current goal is to simulate PSAM (\Cref{sec:psam}).
We now construct a single compute unit that can execute
\emph{any} instruction in a particular PSAM program.
To do so, it suffices to \emph{separately} handle each
instruction of the target program.
This works because the number of instructions in the program is constant,
so we can implement a custom compute unit for each instruction, then glue together a ``top level'' unit that conditionally behaves as any one of these custom units, depending on the PSAM program counter.
This conditional behavior is achieved by executing \emph{each} instruction of the program, then multiplexing the effect of that instruction by using multiplexers controlled by the program counter.
Designing our compute unit this way incurs constant overhead; if one wished to refine our approach, it would be useful to consider carefully designed compute units.

\begin{figure}
  \centering
  \includegraphics[width=\textwidth]{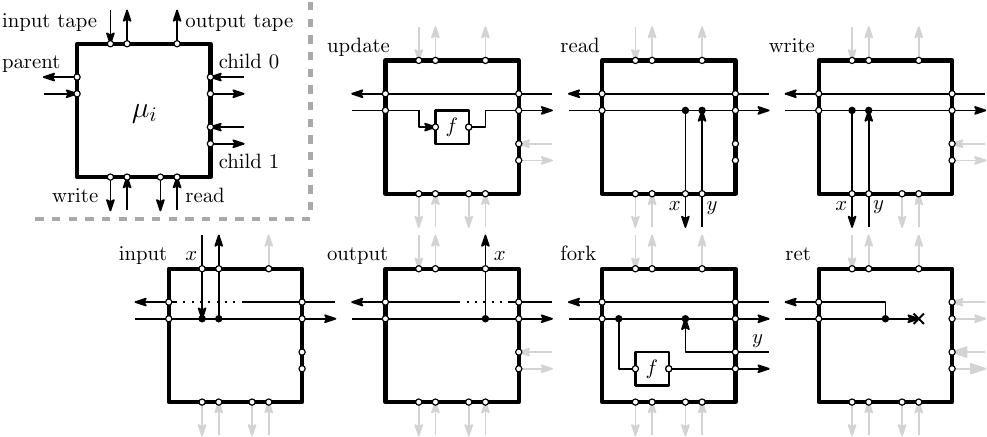}
  \caption{
    The structure of compute units (top left) and the compute unit configuration for each PSAM instruction type.
  }\label{fig:programmed-unit}
\end{figure}

\Cref{fig:programmed-unit} sketches a compute unit for each of the seven PSAM instruction types.
Rather than writing out each of these units as a formal circuit, we find it more direct to simply describe their handling in prose.
In the following, if we do not explicitly mention how the unit handles one of the machine capabilities (receiving an input tape word, writing a word, etc.), then this means that the unit opts out of that capability by sending a $0$ to the appropriate filter.

In each of the following cases except for $\var{fork}$ and $\var{ret}$,
the unit updates its local state, then activates exactly one child by passing the updated state.
Then, the unit receives from its single child a returned word, which it forwards to its own parent (see, e.g., $\var{update}$ in \Cref{fig:programmed-unit}).
Each of the following units appropriately updates the program counter.
Instruction-specific details follow:
\begin{itemize}
  \item $\vb{x} \gets f(\vb{x})$:
    Our first unit describes how to update processor local state.
    This unit receives from its parent a local state $\vb{x}$, then implements $f(\vb{x})$ as a circuit.
  \item $y \gets \var{read}~x$:
    To read a memory word at address $x$, the unit places $x$ on its outgoing read port.
    The memory responds via the incoming read port, which the unit then saves in word $y$ of its local state.
  \item $y \gets \var{write}~x$:
    To write a memory word $x$, the machine places $x$ on its outgoing write port.
    The memory responds with an address, indicating where $x$ was written.
    The unit saves this address in word $y$ of its local state.
  \item $x \gets \var{input}$: To read from the input tape,
    the unit places a one on its outgoing input port, indicating it wishes to read the tape.
    The tape responds on the incoming input port, and the unit saves the value in word $x$ of its local state.
  \item $\var{output}~x$: To write to the output tape,
    the unit stores $x$ on its output port.
  \item $y \gets \var{fork}~f(\vb{x})$:
    To fork a parallel process, the unit first computes $f(\vb{x})$, yielding the forked process starting state.
    Then, the unit sets aside word $y$ of its local state to hold the return pointer of the forked process.
    The unit then activates two children, passing to its first child its local state (with the program counter incremented) and passing to its second child the state $f(\vb{x})$.
  \item $\var{ret}~x$:
    The unit sends local word $x$ to its parent.
    It activates no children.
\end{itemize}

\paragraph{Cleaning up execution.}

At this point, we have shown the core our simulation of PSAM by cyclic circuits.
However, there remains one technical issue:
we have not yet ensured that our cyclic circuit is \emph{legal} (\Cref{defn:legal}).
Each of the individual \emph{components} of our main construction is legal, but we have not yet shown that the composition is legal.

The challenge here is that our circuit may not use up all of its available resources.
For instance, we might not use all provisioned memory reads.
When this happens, there will be portions of partition networks whose routing decisions are unspecified, leading to ambiguity in the circuit's assignment (\Cref{defn:assign}).

The solution to this problem is relatively straightforward: once the PSAM program output is computed, use additional compute units to burn all remaining resources.
These spare units perform the following:
\begin{itemize}
  \item Write all unwritten memory addresses.
  \item Read all unread memory addresses.
  \item Activate all remaining compute units.
  \item Read any remaining words from the input tape.
  \item Write blanks to the end of the output tape.
\end{itemize}
Most of the details here are tedious and unsurprising:
our filters are modified to propagate only the first $n/2$ inputs tagged with $1$;
all subsequent $1$'s are filtered out.
From here, each burn unit just
(1) writes an all zeros word,
(2) reads its own write,
(3) activates two subsequent burn units,
(4) reads an input word,
and (5) writes an all zeros word to the output tape.
We emphasize that the cost here is at most a constant factor blow-up, and
resources can be burned by units that operate in
parallel.

There is one detail that must be explored:
our burn strategy does not explain how to read memory addresses that were written (and not read) by the PSAM program itself.
To burn such values, the machine must be able to \emph{reach} these values.
So far, this is nontrivial.
Indeed, an arbitrary PSAM program could create an `orphan' memory element, for
which there is no sequence of memory addresses the machine can follow to find
the orphan.
In such cases, it is not clear how to read this element, and hence it is not clear how to complete the routing of the memory permutation network.
For this reason, we consider PSAM programs that clean up after themselves:
\begin{definition}[Clean]\label{defn:clean}
  A PSAM program $\mathcal{P}$ is \textbf{clean} if after executing on any input $\vb{x}$,
  $\mathcal{P}$ has not written to any address that has not also been read.
  I.e., $\mathcal{P}$ is $\textbf{clean}$ if each of $\mathcal{P}$'s $\var{write}$ instructions is eventually followed by a $\var{read}$ to the same address.
\end{definition}
Our considered PSAM program simulates PRAM by implementing two binary trees;
to clean up after itself, the program simply traverses those two trees without writing anything back.

When we consider clean PSAM programs, our burn strategy ensures that all partition networks are fully programmed, and hence we achieve a legal cyclic circuit.

\subsection{Uniformity}

Recall that our simulation of cyclic circuits by PRAM only works for circuits
that are polylog uniform.
The fact that our circuit constructions (see
\Cref{sec:circuit-constructions}) are polylog uniformity is trivial from
our explicit descriptions of the circuit's components.
Namely, in our algorithms we can clearly compute tight bounds on the number of gates in a particular subcircuit just by adding a multiplying a polylog number of times.
This makes it easy, based on an integer $i$, to ``search'' the circuit description for the $i$-th gate.

\section{PRAM from PSAM}\label{sec:psam-pram}

In this section, we simulate PRAM with a specific PSAM program.
The goal of this section is to establish the following:
\begin{lemma}[PRAM from PSAM]\label{thm:pram-psam}
  Let $\mathcal{P}$ be a PRAM program (\Cref{sec:word-pram}).
  There exists a PSAM program $\mathcal{P}'$ such that for all $\vb{x}$, $\mathcal{P}(\vb{x}) = \mathcal{P}'(\vb{x})$.
  Moreover, suppose that on a particular length-$n$ input $\vb{x}$,
  $\mathcal{P}(\vb{x})$ halts within $W$ work and $T$ time.
  Then $\mathcal{P}'(\vb{x})$ halts
  within $W \cdot O(\log n)$ work and
  $T \cdot O(\log n)$ time.
\end{lemma}

The proof of \Cref{thm:pram-psam} is by construction of an appropriate PSAM
program, which is described in the remainder of this section.
By combining \Cref{thm:pram-psam} with \Cref{thm:psam-ckt}, we obtain \Cref{thm:main}: an efficient simulation of PRAM by cyclic circuits.

\Cref{sec:psam} showed that PSAM programs can be expressed as
simple recursive procedures over tree structures (\Cref{defn:tree}).
This section presents a number of such procedures that, when composed together, achieve our PRAM simulation.
Recall from \Cref{sec:overview-unit} that our high-level approach is to implement a simple PSAM program that in parallel traverses (1) a memory tree and (2) a processor tree.
This section formalizes that procedure and accounts for its complexity.

\subsection{Tree Operations}

We begin with helper PSAM procedures that manipulate trees.
We give a procedure used to traverse trees, a procedure used to merge trees, and a procedure used to construct fresh trees.

\paragraph{Splitting trees.}
Our first helper PSAM procedure $\var{split}$ is useful for traversing our memory tree.
$\var{split}$ takes as input a pointer to the root of a tree $t$.
It then reads $t$ and splits the read tree into two trees:
\algorithmlineheight
\begin{empheq}[box=\fbox]{align*}
  \gtext{1}~~&\var{split}(t) \defn\\
  \gtext{2}~~&~~\keyword{match}~t~\keyword{with}\\
  \gtext{3}~~&~~~~\var{Empty} \mapsto (\var{Empty}, \var{Empty})\\
  \gtext{4}~~&~~~~\var{Leaf}(x) \mapsto (\var{Leaf}(x), \var{Leaf}(x))\\
\gtext{5}~~&~~~~\var{Branch}(\text{\textvisiblespace}, t^\ell, \text{\textvisiblespace}, t^r) \mapsto (t^\ell, t^r)
\end{empheq}
If $t$ stores an empty or singleton tree, then $\var{split}$
simply outputs two copies of that tree.
If $t$ instead stores an internal tree node, then $\var{split}$ outputs the two subtrees.

$\var{split}$ takes constant time/work on the PSAM, since it simply (1) reads $t$, (2) rearranges words of PSAM local memory, and (3) performs either zero writes (branch case) or two writes (empty and leaf case) to save the two resulting trees.

\paragraph{Merging trees.}
Our second helper procedure takes as input two pointers to trees $t_0$ and $t_1$, and it \emph{merges} those trees (recall discussion in \Cref{sec:overview-unit}).
Namely, the merged tree $t_0 \uplus t_1$ contains all leaves in both $t_0$ and $t_1$, where we merge overlapping leaves with a binary associative operator $\star$.
Recall, we use $\star$ to resolve write conflicts in our CRCW PRAM.
We give a simple example of our merge operation:
\begin{center}
\includegraphics{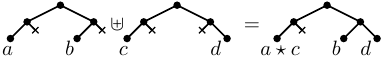}
\end{center}

Our specification of $\uplus$ assumes that all $\var{Leaf}$ nodes are on the same level of each argument tree.
The code is as follows:
\algorithmlineheight
\begin{empheq}[box=\fbox]{align*}
  \gtext{1}~~&t_0 \uplus t_1 \defn\\
  \gtext{2}~~&~~\keyword{match}~(t_0, t_1)~\keyword{with}\\
  \gtext{3}~~&~~~~(\var{Empty}, \text{\textvisiblespace}) \mapsto t_1\\
  \gtext{4}~~&~~~~(\text{\textvisiblespace}, \var{Empty}) \mapsto t_0\\
  \gtext{5}~~&~~~~(\var{Leaf}(x), \var{Leaf}(y)) \mapsto \var{Leaf}(x \star y)\\
  \gtext{6}~~&~~~~(\var{Branch}(d_0^\ell, t_0^\ell, d_0^r, t_0^r), \var{Branch}(d_1^\ell, t_1^\ell, d_1^r, t_1^r) \mapsto \\
  \gtext{7}~~&~~~~~~~~~~~d^\ell \defn \var{max}(d_0^\ell, d_1^\ell)\\
  \gtext{8}~~&~~~~~~~~~~~d^r \defn \var{max}(d_0^r, d_1^r)\\
  \gtext{9}~~&~~~~~~~~~~~t^\ell \defn \keyword{PAR}(t_0^\ell \uplus t_1^\ell)\\
 \gtext{10}~~&~~~~~~~~~~~t^r \defn \keyword{PAR}(t_0^r \uplus t_1^r)\\
 \gtext{11}~~&~~~~~~\keyword{in}~\var{Branch}(d^\ell, t^\ell, d^r, t^r)
\end{empheq}
In words, $\uplus$ proceeds by case analysis.
If either merged tree is empty, then the merge is trivial.
If both trees are singleton, then we combine the content of the leaves with $\star$.
In the general case, we break each tree into two trees, then we pairwise merge via recursion and glue the resulting trees together with a $\var{Branch}$ node.
These cases are exhaustive;
we assumed all leaves reside on the same fixed level, so we will never merge $\var{Branch}$ with $\var{Leaf}$.

Notice that $\uplus$ delegates \emph{each} recursive call to a child process.
This is safe because Line 13 of $\uplus$ simply writes pointers created by the recursive calls, and it does not \emph{dereference} the pointers $t^\ell, t^r$.
Due to the parallel calls, the root pointer of the merged tree is computed in \emph{constant} time on a PSAM, though it can take time linear in the depth of the tree for the PSAM to compute the full merged tree.
The fact that the root is available in constant time allows us to \emph{pipeline} calls to $\uplus$, which will be useful in reducing delay imposed by our PRAM simulation.

If $t_0$ has $n$ total nodes and $t_1$ has $m$ total nodes, then $\uplus$ takes PSAM work at most proportional to the \emph{minimum} of $n$ and $m$.
However, we emphasize that $\uplus$ can in certain cases terminate with \emph{much} less work than this.
For instance, suppose all leaves of $t_0$ are in the left branch of the root, and suppose all leaves of $t_1$ are in the right branch of the root.
In this case, $\uplus$ terminates within \emph{constant} work, because we reach
a trivial merge with $\var{Empty}$ after one recursive call.
Namely, the work required to merge two trees scales only with the amount those trees \emph{overlap} with one another.

\paragraph{Encoding strings as trees.}

Our final helper procedure is used to build new binary trees.
$\var{encode-path}$ takes as input two arguments:
(1) a length-$O(\log n)$ binary string $\vb{i}$ that represents an index and
(2) a data value $x$ to store.
$\var{encode-path}$ interprets $\vb{i}$ as the name of a path through a tree.
It then builds a tree with a single leaf that stores $x$ and that lies along path $\vb{i}$.
See the following example:
\begin{center}
\includegraphics{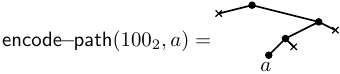}
\end{center}

The specification for $\var{encode-path}$ follows:
\algorithmlineheight
\begin{empheq}[box=\fbox]{align*}
  \gtext{1}~~&\var{encode-path}(\vb{i}, x) \defn\\
  \gtext{2}~~&~~\keyword{match}~\vb{i}~\keyword{with}\\
  \gtext{3}~~&~~~~[~] \mapsto \var{Leaf}(x)\\
  \gtext{4}~~&~~~~[0] \concat \vb{i}' \mapsto \var{Branch}(|\vb{i}'|, \keyword{PAR}(\var{encode-path}(\vb{i}', x)), 0, \var{Empty})\\
  \gtext{5}~~&~~~~[1] \concat \vb{i}' \mapsto \var{Branch}(0, \var{Empty}, |\vb{i}'|, \keyword{PAR}(\var{encode-path}(\vb{i}', x)))
\end{empheq}
Note that because $\vb{i}$ is a length-$O(\log n)$ binary string, we can
represent $\vb{i}$ as a single word of PSAM local memory.
$\var{encode-path}$ proceeds by case analysis on $\vb{i}$,
at each step branching either left or right.

As in our code for $\uplus$, $\var{encode-path}$ recurses in parallel.
This ensures that the pointer to the root of the encoding tree is available within constant time,
allowing pipelining of procedures.
It is easy to see that $\var{encode-path}$ fully computes the resulting tree in PSAM time/work proportional to the length of $\vb{i}$.

\subsection{Simulating PRAM}

\begin{figure}
{\small
\algorithmlineheight
\begin{empheq}[box=\fbox]{align*}
  \gtext{1}~~&\var{step}(\var{memory}, \var{processors}) \defn\\
  \gtext{2}~~&~~\keyword{match}~\var{processors}~\keyword{with}\\
  \gtext{3}~~&~~~~\var{Empty}\mapsto (\var{memory}, \var{Empty})\\
  \gtext{4}~~&~~~~\var{Branch}(d^\ell, p^\ell, d^r, p^r) \mapsto\\
  \gtext{5}~~&~~~~~~\keyword{let}~(m^\ell, m^r) \defn \var{split}(\var{memory})\\
  \gtext{6}~~&~~~~~~~~~~~(m_{\var{shallow}}, p_{\var{shallow}}) \defn \keyword{if}~d^\ell \leq d^r\\
  \gtext{7}~~&~~~~~~~~~~~~~~~~~~~~~~~~~~~~~~~~~~~~~\keyword{then}~\keyword{PAR}(\var{step}(m^\ell, p^\ell))\\
  \gtext{8}~~&~~~~~~~~~~~~~~~~~~~~~~~~~~~~~~~~~~~~~\keyword{else}~\keyword{PAR}(\var{step}(m^r, p^r))\\
  \gtext{9}~~&~~~~~~~~~~~(m_{\var{deep}}, p_{\var{deep}}) \defn \keyword{if}~d^\ell \leq d^r\\
 \gtext{10}~~&~~~~~~~~~~~~~~~~~~~~~~~~~~~~~~~~\keyword{then}~\var{step}(m^r, p^r)\\
 \gtext{11}~~&~~~~~~~~~~~~~~~~~~~~~~~~~~~~~~~~\keyword{else}~\var{step}(m^\ell, p^\ell)\\
 \gtext{12}~~&~~~~~~\keyword{in}~(m_{\var{shallow}} \uplus m_{\var{deep}}, p_{\var{shallow}} \uplus p_{\var{deep}})\\
 \gtext{13}~~&~~~~\var{Leaf}(\rho) \mapsto\\
 \gtext{14}~~&~~~~~~\keyword{let}~\var{val}\defn \keyword{match}~\var{memory}~\keyword{with}\\
 \gtext{15}~~&~~~~~~~~~~~~~~~~~~~~~~~\var{Leaf(x)} \mapsto x\\
 \gtext{16}~~&~~~~~~~~~~~~~~~~~~~~~~~\var{Empty} \mapsto 0^w\\
 \gtext{17}~~&~~~~~~~~~~~(\var{val}', \rho_0^?, \rho_1^?) \defn \var{handle-pram-instruction}(\rho, \var{val})\\
 \gtext{18}~~&~~~~~~~~~~~p_0 \defn \var{encode-possible-processor-state}(\rho_0^?)\\
 \gtext{19}~~&~~~~~~~~~~~p_1 \defn \var{encode-possible-processor-state}(\rho_1^?)\\
 \gtext{20}~~&~~~~~~\keyword{in}~(\var{encode-path}(\var{addr}(\rho), \var{val}'), p_0 \uplus p_1)\\
 \gtext{21}~~&\\
 \gtext{22}~~&\var{encode-possible-processor-state}(\rho^?) \defn\\
 \gtext{23}~~&~~\keyword{match}~\var{\rho^?}~\keyword{with}\\
 \gtext{24}~~&~~~~\var{Empty} \mapsto \var{Empty}\\
 \gtext{25}~~&~~~~\var{Leaf}(\rho) \mapsto \var{encode-path}(\var{addr}(\rho) \concat \var{proc-id}(\rho), \rho)
\end{empheq}
}
\caption{
  Our PSAM program that simulates a single step of a PRAM program. 
  $\var{step}$ takes as input a binary tree encoding of the current PRAM
  memory/current PRAM processors (see \Cref{inv:main}), and it traverses these trees simultaneously,
  pairing each processor state with the appropriate memory element.
  This allows each PRAM processor to perform one instruction.
  As $\var{step}$'s recursion unwinds, it builds a fresh
  memory/processor tree, ready for the next PRAM step.
}\label{fig:step}
\end{figure}

Now that we have built our PSAM procedures for manipulating trees,
we are ready to simulate PRAM.
Our simulation maintains two binary trees:
(1) a \emph{memory} tree, which holds written memory values, and
(2) a \emph{processor} tree, which stores the state of each active processor.
Our simulation enforces the following invariant on the structure of these trees:
\begin{invariant}[PRAM-by-PSAM Invariant]\label{inv:main}
  Our PRAM simulation at each step maintains two binary trees $\var{memory}$ and $\var{processors}$ whose content is as follows:
  \begin{align*}
    \var{memory} &= \biguplus_{(i, \var{val}) \in \var{mem}} \var{encode-path}(\var{bin}(i), \var{val})\\
    \var{processors} &= \biguplus_{\rho \in \var{procs}} \var{encode-path}(
      \var{addr}(\rho)\concat\var{proc-id}(\rho), \rho)
  \end{align*}
  Here, $\var{mem}$ denotes a partial map from addresses $i$ to values $\var{val}$ written to PRAM memory so far,
$\var{procs}$ denotes the set of active PRAM processor states,
$\var{bin}(i)$ denotes a binary encoding of $i$,
$\var{addr}(\rho)$ denotes the binary memory address requested by processor $\rho$, and $\var{proc-id}(\rho)$ denotes a unique ID assigned to processor $\rho$.
\end{invariant}

The memory tree is the merge of all values written to memory so far, where each
memory element $\var{val}$ is placed at a leaf corresponding to its index $i$.
Importantly, our memory tree only stores addresses that have been written so far;
all unwritten addresses are implicitly assumed to hold the all zeros string, and we do not need to explicitly represent those values in the tree. 
This fact ensures that the tree size is at most linear in the work of the PRAM program.

The processor tree is the merge of trees encoding processor states;
each processor state $\rho$ is placed at a leaf corresponding to the memory address $\var{addr}(\rho)$ that the processor wishes to access.
We ensure each processor $\rho$ has a distinct processor ID $\var{proc-id}(\rho)$, and more precisely
each processor state $\rho$ is placed at leaf $\var{addr}(\rho) \concat \var{proc-id}(\rho)$.
IDs ensure that no two processor states reside on the same path.

\Cref{fig:step} lists our PRAM-by-PSAM simulation.
Our simulation assumes existence of a procedure $\var{handle-pram-instruction}$.
$\var{handle-pram-instruction}$ takes as input (1) a processor state and (2) the memory value that processor currently wishes to access.
It outputs
(1) a value\footnote{
  If the processor wishes to merely read, it can write back whatever it just read; $\star$ should be properly chosen to prefer writes over reads.
} the processor wishes to write back to memory and
(2) two optional processor states.
I.e., the process can (1) halt by returning two empty processor states,
(2) continue by returning one non-empty processor state, or
(3) fork by returning two non-empty processor states.
For simplicity, we assume that these optional processor states are encoded as either (1) an empty tree or (2) a singleton tree.
We assume the details of handling an actual PRAM instruction, e.g. adding/multiplying/comparing/branching, are handled by $\var{handle-pram-instruction}$.
We also assume that $\var{handle-pram-instruction}$ runs in some fixed number of PSAM instructions that is at most $O(\log n)$.
Parameterizing over $\var{handle-pram-instruction}$ provides flexibility in the capabilities of the PRAM. 
We emphasize that all the difficult parts of simulating PRAM, e.g. routing memory and appropriately parallelizing behavior, are handled by $\var{step}$.

$\var{step}$ takes as input two pointers to trees: $\var{memory}$ and $\var{processors}$.
These trees inductively satisfy \Cref{inv:main}.
$\var{step}$ returns as output a pointer to a fresh memory tree and a fresh processor tree which also satisfy \Cref{inv:main}.
In this manner, $\var{step}$ can be repeatedly applied.
To achieve full-fledged PRAM, we compute a fixed-point by applying $\var{step}$ until the processor tree is empty.

$\var{step}$ proceeds by reading $\var{processors}$ and then case analyzing the tree.
If the processor tree is empty, then no work needs to be done.
Other cases are more detailed.

\textbf{If the processor tree is an internal node},
then we recursively handle both subtrees.
More specifically, we first $\var{split}$ the memory tree,
then recurse on the left trees/right subtrees.
This handling ensures that we can indeed match each processor with its desired memory element, because \Cref{inv:main} ensures that each processor state is in a subtree that matches the subtree of the memory element it wishes to read.

As an important detail, our recursion proceeds in parallel, delegating the shallower processor tree to a child PSAM process.
This ensures that the child process terminates before the parent process completes its own recursive call.
This, in turn, ensures that the dereference of $m_{\var{shallow}}, p_{\var{shallow}}$
(line 12) does not illegally dereference the return pointer of an incomplete child
process (see \Cref{sec:psam}).

Once both recursive calls are complete, we have two memory trees and two processor trees, which we pairwise merge with calls to $\uplus$.

\textbf{If the processor tree is a leaf node}, then we are ready to dispatch a single PRAM instruction.
To do so, we first case analyze the $\var{memory}$ tree.
Note that it is impossible that the memory tree is a $\var{Branch}$ node,
because by \Cref{inv:main} the processor tree is always deeper than the memory
tree.
Thus, the memory tree can either hold the single RAM element required by the current processor $\rho$, or it can be empty.
The latter case corresponds to a case where the processor $\rho$ wishes to access a memory cell that has not yet been accessed by any processor.
In this latter case, the processor will read the all zeros word.

From here, we call $\var{handle-pram-instruction}$ with the processor state and the accessed data element.
The processor returns an element to be written back to memory and up to two subsequent processor states.
We encode these values as trees, then we return the fresh memory/processor trees.

\paragraph{Cleaning up.}
Once we compute the fixed point of $\var{step}$, the PRAM simulation is finished.
However, recall that to achieve a \emph{clean} PSAM procedure (\Cref{defn:clean}), our program must read all unread data.

Cleaning up is straightforward.
Our processor tree is, by the fact that we have finished the simulation,
trivially empty, and to clean the memory tree, we simply traverse it from root
to leaves in a recursive procedure, without writing back any data.
This cleanup imposes at most constant factor overhead, as the PRAM memory cannot have accessed more than $W(n)$ leaves.

\subsection{Step Complexity}

In the following, we assume that the processor tree has $p$ leaves; i.e., there are $p$ active processors.

\paragraph{Work.}
We claim that $\var{step}$ uses at most $p \cdot O(\log n)$ total PSAM work.
To see this, first note that both the memory tree and the processor tree have $O(\log n)$ depth; this is ensured by the fact that the number of processors and the highest memory address are each at most polynomial in $n$.
If we for now ignore calls to $\uplus$, we see that each time we call $\var{step}$ on an internal node,
we expend constant work (and time) before recursing down both subtrees.
Our recursion only follows paths through the tree to locations where processors reside,
so the total number of recursive calls is at most $p \cdot O(\log n)$.

The calls to $\uplus$ are more difficult to analyze, but they similarly can be seen to consume at most $p \cdot O(\log n)$ total work.
One way to see this is as follows:
For the processor tree, we merge together $O(p)$ trees resulting from calls to $\var{encode-path}$, where each merged tree is a ``path tree'' with one leaf and $O(\log n)$ internal nodes.
Our $\var{step}$ procedure merges these trees together in some order resulting from $\var{steps}$'s recursion tree.
However, instead suppose we were to merge these trees together one at a time, maintaining an accumulator tree and merging each path tree into this accumulator.
Each such merge would cost at most $O(\log n)$ work, simply due to the small size of the path tree.
Thus, this strategy uses at most $p \cdot O(\log n)$ work.

Now, $\var{step}$ does not merge path trees into an accumulator in this way; it merges trees at each recursive call site.
However, we claim that $\var{step}$ uses \emph{strictly less work} than the above accumulator approach.
Indeed, the merge of two trees is strictly smaller than those two trees alone, since the merged tree will join together some internal nodes.
Hence, merging the processor tree consumes at most $p \cdot O(\log n)$ work.
This same argument holds for the memory tree.

Hence, $\var{step}$ consumes at most $p \cdot O(\log n)$ work.

\paragraph{Time.}

We claim that $\var{step}$ computes root pointers of its output trees within $O(\log n)$ PSAM time.
This follows almost immediately from the fact that both the memory and processor tree have logarithmic depth.

There is one non-trivial component to this argument: we must argue that calls to $\uplus$ can be pipelined.
Indeed, recall that $\uplus$ takes $O(\log n)$ to completely compute its output, but it computes its root pointer  within \emph{constant} time.

Formally, we can establish an inductive hypothesis that $\var{step}$ outputs its root pointers within $O(\log n)$ time.
In the base case this holds, since $\var{handle-pram-instruction}$ is
assumed to take at most $O(\log n)$ time and since it takes $O(\log n)$ steps
to encode a tree path (indeed, the root of the encoded path is available within $O(1)$ time).
In the inductive case, we recursively call $\var{step}$ twice in parallel, taking $O(\log n)$ time by the inductive hypothesis.
We conclude with calls to $\uplus$, which return their root in $O(1)$ time.
Hence, the general case incurs only constant additive time on top of its recursive calls.
As the depth of the tree is at most logarithmic, the induction holds.
Hence, $\var{step}$ computes its output within $O(\log n)$ time.

\section{Concluding Remarks}

By combining our results from \Cref{sec:circuit-constructions,sec:psam-pram}, we attain
a simulation of CRCW PRAM, and our simulation is achieved using cyclic circuits, a simple and easily implemented model.
Surprisingly, the simulation incurs only polylogarithmic overhead in terms of both work and time.

This demonstrates feasibility of powerful parallel machines, at least in theory.
Of course, substantial effort is needed to determine if this theoretical feasibility can translate to practical outcomes.
At the least, we believe our result shows that cyclic circuits are far more interesting than previously thought.
While prior works investigated cyclic circuits, none showed a connection between this simple model and PRAM.

\paragraph{An open question.}
Can our simulation of PRAM by cyclic circuits be improved?
Our simulation achieves $O(\log^4 n)$ work overhead and $O(\log^3 n)$ runtime overhead, and it is not obvious how to do significantly better.
We incur $O(\log n)$ overhead from word size, $O(\log^2 n)$ overhead from the dynamic permutation network, and $O(\log n)$ overhead from the simulation of PRAM by PSAM.
Thus, improving our result in a modular way would require either (1) changing the circuit model, (2) improving dynamic permutation networks, or (3) improving the simulation of PRAM by PSAM;
none of these improvements are obvious.
Of course, it might also be possible to mix the concerns of the dynamic permutation with PRAM-by-PSAM simulation, or by applying some completely different approach; these ideas are also not clear.

In terms of negative results, our only current insight is that there is a trivial $\Omega(\log n)$ lower bound on work overhead.
This bound comes simply from ``unfairness'' of the comparison between the considered PRAM and cyclic circuits: circuit wires hold individual bits while PRAM manipulates $\Theta(\log n)$-bit words.
To show a $\Omega(\log n)$ lower bound, we can consider simulating a PRAM program that, e.g., forces the cyclic circuit to take as input $O(n)$ elements and write them out in some input-specified order.
To achieve generality, the circuit must fully read the $O(n\cdot \log n)$-bit input.

One might expect that this permutation problem would allow us to strengthen the lower bound to $\Omega(\log^2 n)$ work overhead, since it is well known that comparison-based sorting requires $\Theta(n \cdot \log n)$ comparisons. 
However, trying to apply this insight immediately runs into well-known open questions regarding the complexity of sorting, see e.g.~\cite{FarHajMohKasShi19,AshLinWeiShi22}.
Namely, it might be possible to sort with $o(n \cdot \log n)$ Boolean gates, and this prevents us from establishing this stronger lower bound.

\bibliographystyle{alpha}
\bibliography{bib}

\end{document}